\newif\ifDraft
\Drafttrue

\newif\ifArxiv
\Arxivtrue

\newif\ifShowFigures
\ShowFigurestrue

\documentclass[smallextended]{svjour3}   
\smartqed  
\ifArxiv
\def\makeheadbox{{%
\hbox to0pt{\vbox{\baselineskip=10dd\hrule\hbox
to\hsize{\vrule\kern3pt\vbox{\kern3pt
\hbox{\bfseries Arxiv Technical Report\ }
\kern3pt}\hfil\kern3pt\vrule}\hrule}%
\hss}}}
\usepackage{hyperref}
\else
\journalname{Algorithmica}
\fi
\usepackage{graphicx}
\usepackage{latexsym}
\usepackage{color}
\usepackage{multicol}
\usepackage{url}
\usepackage{verbatim}
\usepackage{txfonts}
\usepackage{enumerate}

\graphicspath{{BLAH/}{FIGURESFull/}{FIGURES/}{./}}

\ifDraft
\newcounter{mycommentcounter}
\newcommand{\Comment}[2][Comment]{\refstepcounter{mycommentcounter}%
    \ifhmode%
     \unskip%
     {\dimen1=\baselineskip \divide\dimen1 by 2 %
       \raise\dimen1\llap{\tiny
	{-\themycommentcounter-}}}\fi%
     \marginpar[{\renewcommand{\baselinestretch}{0.8}%
       \hspace*{-1em}\begin{minipage}{6.6em}\footnotesize%
[\themycommentcounter]:%
\raggedright \underline{#1}: #2\end{minipage}}]{\renewcommand{\baselinestretch}{0.8}%
       \begin{minipage}{12em}\footnotesize%
[\themycommentcounter]: \raggedright%
\underline{#1}: #2\end{minipage}}%
}
\else
\newcommand{\Comment}[2][Comment]{\relax}        
\fi

\newcommand{\CR}{\ensuremath{{\cal R}}}
\newcommand{\highlight}[1]{{\bfseries\itshape #1}}
\newcommand{\define}[1]{{\em #1}}

\begin{document}
\title{Optimal Polygonal Representation of Planar Graphs\thanks{A preliminary version of this paper appeared in LATIN 2010, Oaxaca, Mexico}
}

\author{
C.~A.~Duncan \and E.~R.~Gansner \and Y.~F.~Hu  \and M.~Kaufmann \and S.~G.~Kobourov
}

\institute{
    C.~A.~Duncan \at
    Dept. of Computer Science, Louisiana Tech University\\
    \email{duncan@latech.edu}
  \and
    E.~R.~Gansner \and Y.~F.~Hu \at
    AT\&T Research Labs, Florham Park, NJ\\
    \email{\{erg, yifanhu\}@research.att.com}
  \and
    M.~Kaufmann \at
    Wilhelm-Schickhard-Institut for Computer Science, T\"ubingen University\\
    \email{mk@informatik.uni-tuebingen.de}
  \and
    S.~G.~Kobourov \at
    Dept. of Computer Science, University of Arizona\\
    Tel.: +123-45-678910\\
    Fax: +123-45-678910\\
    \email{kobourov@cs.arizona.edu}
}

\ifArxiv
\date{Received: April, 2011}
\fi

\maketitle

\begin{abstract}
In this paper, we consider the problem of representing graphs by polygons whose sides touch.
We show that at least six sides per polygon are necessary by constructing a class of planar graphs that cannot be represented by pentagons.
We also show that the lower bound of six sides is matched by an upper bound of six sides with a linear-time algorithm for representing any planar graph by touching hexagons.
Moreover, our algorithm produces convex polygons with edges having at most three slopes 
and with all vertices lying on an $O(n) \times O(n)$ grid.
\keywords{Planar graphs \and Contact graphs 
\and Graph drawing \and Polygonal drawings
}
\end{abstract}

\section{Introduction}
For both theoretical and practical reasons, there is a large body of work considering how to
represent planar graphs as {\em contact graphs}, i.e., graphs whose vertices are represented by
geometrical objects with edges corresponding to two objects touching in some specified fashion.
Typical classes of objects might be curves, line segments or isothetic rectangles, and 
an early result is Koebe's theorem~\cite{Koebe36}, which shows that all planar graphs can be
represented by touching disks. 

In this paper, we consider contact graphs whose objects are simple polygons, with an edge occurring
whenever two polygons have non-trivially overlapping sides. As with treemaps~\cite{Bruls+Huizing+Wijk+2000a}, such representations
are preferred in some contexts~\cite{Buchsbaum08} over the standard node-link representations for displaying relational information.
Using adjacency to represent a connection can be much more compelling, and cleaner, than drawing a line segment between two
nodes. For ordinary users, this representation suggests the familiar metaphor of a geographical map.
 
It is clear that any graph represented this way must be planar.
As noted by de Fraysseix {\em et al.}~\cite{FMR04},
it is also easy to see that all planar graphs have such representations for sufficiently
general polygons.
Starting with a straight-line planar drawing of a graph, 
we can create a polygon for each vertex by taking the midpoints of all adjacent edges and the centers of all neighboring faces. 
Note that the number of sides in each such polygon is proportional to the degree of its vertex. 
Moreover, these polygons are not necessarily convex; see Figure~\ref{fig-map}.

It is desirable, for aesthetic, practical and cognitive reasons, to limit the complexity of the polygons involved, where ``complexity'' here means the number of sides in the polygon.
Fewer sides, as well as wider angles in the polygons, make for simpler and cleaner drawings. In related applications such as 
floor-planning~\cite{Liao03}, physical constraints make polygons with very small angles or many sides undesirable.
One is then led to consider how simple such representations can be. 
How many sides do we really need?
Can we insist that the polygons be convex, perhaps with a lower bound on the size of the angles or the edges?
If limiting some of these parameters prevents the drawing of all planar graphs, which ones can be drawn?

\begin{figure}
  \begin{center}
\ifShowFigures
\begin{tabular}{ccc}
    \includegraphics[scale=.5]{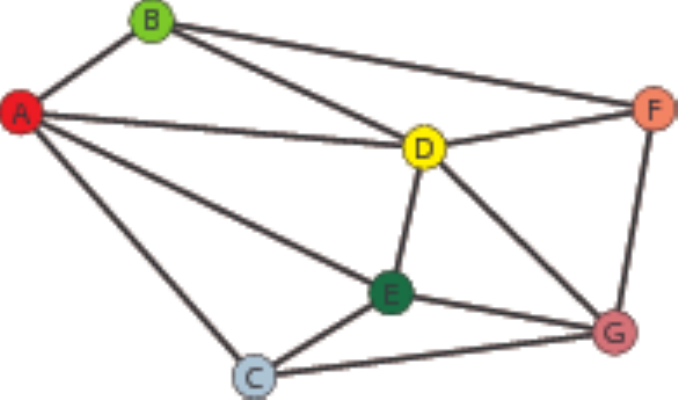} &
    \includegraphics[scale=.5]{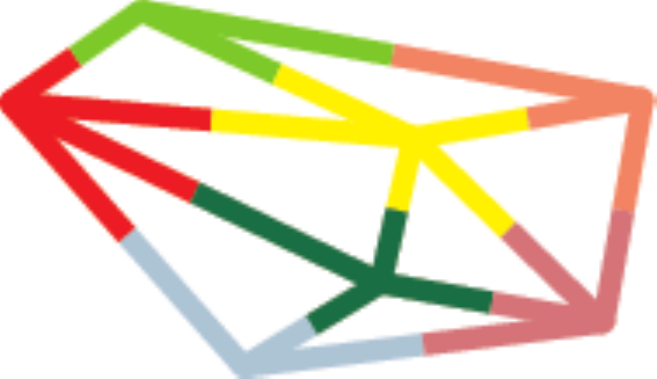} &
    \includegraphics[scale=.5]{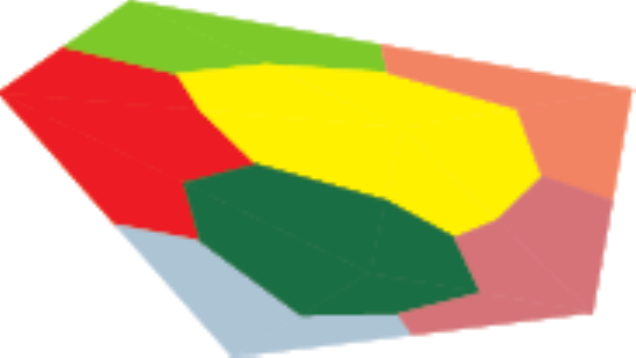}\\
    (a) & (b) & (c)
\end{tabular}
\else
    Figure omitted because causing problems with my printer!
\fi
  \end{center}
  \caption{\small\sf\label{f:thickedges} (a) A drawing of a planar graph.
(b) We apportion the edges to the endpoints by cutting each 
edge in half.
(c) We then apportion the faces to form polygons.}
\label{fig-map}
\end{figure}

\vspace{-0.2in}
\subsection{Our Contribution}
This paper provides answers to some of these questions. 
Previously, it was known~\cite{He99,Liao03} that triangulated planar graphs can be represented using non-convex octagons.
On the other hand, it is not hard to see that one cannot use 
triangles (e.g., $K_5$ minus one edge cannot be represented with triangles~\cite{ttg}). 

Our main result is showing that hexagons are necessary and sufficient
for representing all planar graphs.
For necessity we construct a class of graphs that cannot be represented using five or fewer sides. 
For sufficiency, we prove the following:
\begin{theorem}
\label{thm:t6g}
For any planar graph $G$ on $n$ vertices, we can construct 
in linear time on an $O(n) \times O(n)$ grid
a touching hexagons representation of $G$ with convex regions.
Moreover, if the graph is a triangulation, the representation is
also a tiling.
\end{theorem}
Note, if the input graph is not triangulated, there might be convex holes.
We, in fact, prove this theorem using two different methods.
First, in Sections~\ref{t6g} and~\ref{sec_compact}, we describe a linear-time 
algorithm that
produces a representation using convex hexagons 
along with a linear-time compaction algorithm
to reduce the initial exponential area to an  $O(n) \times O(n)$ integer grid.
Second, in Section~\ref{sec_alternate},
we show how modifying Kant's algorithm for hexagonal grid drawings of 
3-connected, 3-regular planar
graphs~\cite{kant-92} produces a similar result by different means.
In both variations, the drawings use at most three slopes for the sides, 
for example, 1, 0 and -1.

\vspace{-0.2in}
\subsection{Related Work}
As remarked above, there is a rich literature related to various types of contact graphs. 
There are many results considering curves and line segments as objects (cf.~\cite{ccg:h98,contact:hk01}).
For closed shapes such as polygons, results are rarer, except for axis-aligned (or {\em isothetic}) rectangles. 
In a sense, results on representing planar graphs as ``contact systems'' can be dated back to 
Koebe's 1936 theorem~\cite{Koebe36} which states that any planar graph can be represented as a contact 
graph of disks in the plane.

The focus of this paper is side-to-side contact of polygons.
The algorithms of He~\cite{He99} and Liao {\em et al.}~\cite{Liao03} 
produce contact graphs of this type
for triangulated graphs, with nodes represented by the
union of at most two isothetic rectangles, thus giving a polygonal representation by non-convex octagons.

We now turn to contact graphs using isothetic rectangles, which are often referred 
to as {\em rectangular layouts}.
This is the most extensively studied class of contact graphs, due in part
to its relation to application areas such as VLSI floor-planning~\cite{fp:ll88,slice:ys95}, architectural design~\cite{arch:s76} and
geographic information systems~\cite{rdgis:gs69}, but also due to
the mathematical ramifications and connections to
other areas such as
rectangle-of-influence drawings~\cite{rid:llmw98} and proximity drawings~\cite{prox:dll94,rng:jt92}.

Graphs allowing rectangular layouts have been fully 
characterized~\cite{Rahman04,rl:t86} with linear algorithms for deciding
if a rectangular layout is possible and, if so, constructing one. The simplest formulation~\cite{Buchsbaum08} notes
that a graph has a rectangular layout if and only if it has a planar embedding with no filled triangles.
Thus, $K_4$ has no rectangular layout. Buchsbaum {\em et al.}~\cite{Buchsbaum08} also show, using results of
Biedl {\em et al.}~\cite{rid:bbm99}, that graphs that admit rectangular layouts
are precisely those that admit a weaker variation of planar rectangle-of-influence drawings.

Rectangular layouts required to form a partition of a rectangle are known as {\em rectangular duals}.
In a sense, these are ``maximal'' rectangular layouts; many of the results concerning rectangular layouts
are built on results concerning rectangular duals.
Graphs admitting rectangular duals have been characterized~\cite{He93,rdual:kk88,rd:ll90} and there are linear-time 
algorithms~\cite{He93,rel:kh97} for constructing them.

%
%

Another view of rectangular layouts arises in VLSI floor-planning, where a rectangle is partitioned 
into rectilinear regions so that region adjacencies correspond to a given planar graph. 
It is natural to try to minimize the complexities of the resulting regions. The best known results are 
due to He~\cite{He99} and Liao {\em et al.}~\cite{Liao03} who show that regions need not have more than 8 sides. 
Both of these algorithms run in $O(n)$ time and produce layouts on an $O(n) \times O(n)$ integer grid
where $n$ is the number of vertices.

Rectilinear cartograms can be defined as rectilinear contact graphs for vertex-weighted planar graphs, where the 
area of a rectilinear region must be proportional to the weight of its corresponding node. Even with this extra 
condition, de Berg {\em et al.}~\cite{deBerg07} show that rectilinear cartograms can
always be constructed in $O(n \log n)$ time, using regions having at most 40 sides. The resulting regions, however, are highly non-convex and can have poor aspect ratio.
Recently, Alam {\em et al.}~\cite{abfkk-pcrpg-11} explore
optimal bounds on the complexity of the polygons needed to produce point-contact and side-contact
representations of subclasses of vertex-weighted 
planar graphs with additional restrictions such as convexity and hole-free regions.

Although not considered by the authors, an upper bound of six for the minimum 
number of sides in a touching polygon representation of planar graphs might 
be obtained from the vertex-to-side triangle contact graphs of 
de Fraysseix {\em et al.}~\cite{FMR04}. 
The top edge of each triangle can be converted into a raised 3-segment polyline,
clipping the tips of the triangles touching it from above, thereby turning the triangles into side-touching hexagons. 
This approach might prove difficult
for generating hexagonal representations as it involves 
computing the amounts by which each triangle may be raised so as to become 
a hexagon without changing any of the adjacencies. 
Moreover, the nature of 
such an algorithm would produce many ``holes,'' potentially making such 
drawings less appealing, or requiring further modifications.
In~\cite{glp-tcrd-10}, Gon\c{c}alves {\em et al.} describe a similar approach after
presenting an algorithm to create primal-dual triangle 
contact representations, where each node and face are represented as triangles.

\subsection{Preliminaries}

\highlight{Touching Hexagons Graph Representation:} Throughout this paper, we assume we are dealing with a connected planar graph $G=(V,E)$. We would like to construct a set of closed simple polygons $R$ whose interiors are pairwise disjoint, along with an isomorphism $\CR : V \rightarrow R$, such that for any two vertices $u,v \in V$, the boundaries of $\CR(u)$ and $\CR(v)$ overlap non-trivially if and only if $\{u,v\}\in E$. For simplicity, we adopt a convention of the cartogram community and define the \define{complexity}
of a polygonal region as the number of sides it has. We call the set of all graphs having such a representation where each polygon in $R$ has complexity 6 \define{touching hexagons graphs}.

\smallskip\noindent\highlight{Canonical Labeling:} 
Our algorithms begin by first computing a \define{planar embedding} of the input graph $G=(V,E)$ and using 
that to obtain a \define{canonical labeling} of the vertices.
A planar embedding of a graph is simply a clockwise order of the neighbors of each vertex in the graph. 
Obtaining a planar embedding can be done in linear time using the algorithm by Hopcroft and Tarjan~\cite{ht-ept-74}. 
The canonical labeling or order of the vertices of a planar graph was defined by de Fraysseix {\em et al.}~\cite{fpp-hdpgg-90} in the context of straight-line drawings of planar graphs on an integer grid of size $O(n) \times O(n)$.
While the first algorithm for computing canonical orders required $O(n \log n)$ time~\cite{fpp-sssfe-88}, Chrobak and Payne~\cite{cp-ltadp-95} have shown that this can be done in $O(n)$ time.

In this section we review the canonical labeling of a planar graph as defined by de Fraysseix {\em et al.}~\cite{fpp-sssfe-88}. Let $G=(V,E)$ be a fully triangulated planar graph embedded in the plane with exterior face $u,v,w$. A canonical labeling of the vertices $v_0=u, v_1=v, v_2, \dots, v_{n-1}=w$ is one that meets the following criteria for every $2 < i < n$:
\begin{enumerate}
\item The subgraph $G_{i-1}\subseteq G$ induced by $v_0, v_1, \dots, v_{i-1}$ is 2-connected, and the boundary of its outer face is a cycle $C_{i-1}$ containing the edge $(u,v)$;
\item The vertex $v_i$ is in the exterior face of $G_{i-1}$, and its neighbors in $G_{i-1}$ form an (at least 2-element) subinterval of the path $C_{i-1}-(u,v)$.
\end{enumerate}

The canonical labeling of a planar graph $G$ allows for the incremental placement of the vertices of $G$ on a grid of size $O(n) \times O(n)$ so that when the edges are drawn as straight-line segments there are no crossings in the drawing. The two criteria that define a canonical labeling are crucial for the region creation step of our algorithm.

Kant generalized the definition for triconnected graphs,
partitioning the vertices into sets $V_1$ to $V_K$
that can be either singleton vertices or chains of vertices~\cite{k-dpguc-96}.

\section{Lower Bound of Six Sides}
\label{sec:opt}

In this section 
we show that at least six sides per polygon are sometimes needed in 
a touching polygons representation of a planar graph. 
We begin by constructing a class of planar graphs that cannot be represented by four-sided polygons and then extend the argument to show that the class also cannot be represented by five-sided polygons. 

\vspace{-0.2in}
\subsection{Four Sides Are Not Enough}
Consider the fully triangulated graph $G^k$ in Figure~\ref{fig-michael}(a).
It has three nodes on the outer face $A, B$ and $C$, and contains a chain of nodes $1,\dots,k$ which are all adjacent to $A$ and $B$.
Consecutive nodes in the chain, $i$ and $i+1$, are also adjacent. The remaining nodes of $G^k$ are degree-3 nodes $l_i$ and $r_i$ inside the triangles $\Delta(A,i,i+1)$ and $\Delta(B,i,i+1)$.

\begin{figure}[t]
\begin{center}
\includegraphics[width=.7\textwidth]{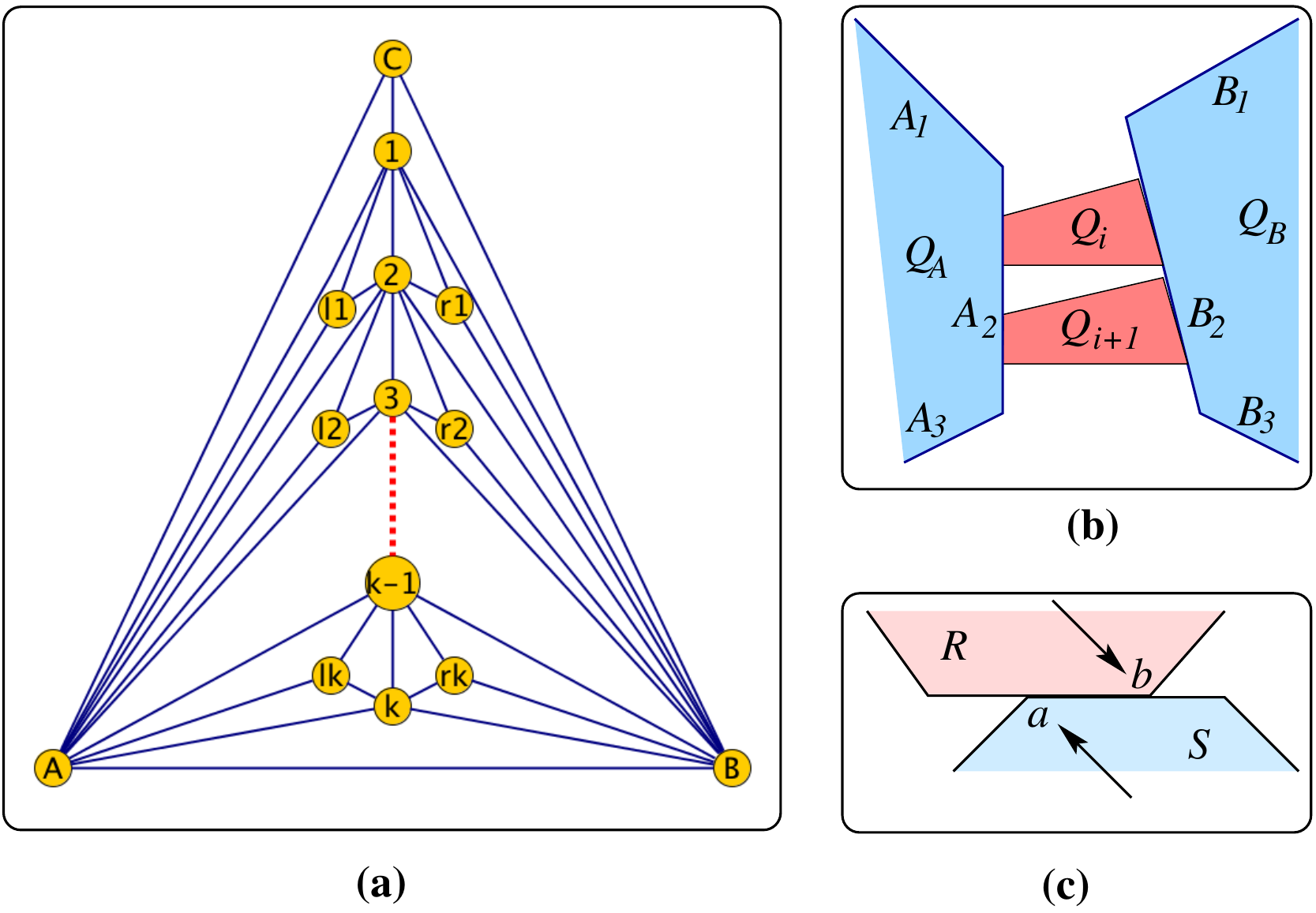}
\caption{\small\sf (a)  The graph $G^k$ that provides the counterexample. (b) A pair of subsequent fair quadrilaterals adjacent to the same sides of $Q_A$ and $Q_B$. (c) Illustration for Observation~2 shows one of  three possible cases for two touching regions.
\label{fig-michael}
}
\end{center}
\end{figure}



\begin{theorem}
\label{thm:moreThan4}
For $k \geq 33$, there does not exist a touching polygons representation for $G^k$ in which all regions have complexity four or less.
\end{theorem}

\begin{proof}
Assume, for the sake of contradiction, that we are given a touching polygons drawing for $G^k$ in which all regions have complexity four or less.
Without loss of generality, we assume that the drawing has an embedding that corresponds to the one shown in Figure~\ref{fig-michael}(a).
Let $Q_A$, $Q_B$ and $Q_C$ 
denote the quadrilaterals representing nodes $A$, $B$ and $C$, and 
let $Q_i$ denote the quadrilateral representing node $i$.
Once again, without loss of generality, let $Q_A$ lie in the left corner, $Q_B$ in the right corner and $Q_C$ at the top of the drawing.

We start with an observation.

{\bf Observation 1:} 
Any corner of a quadrilateral can be adjacent to at most two disjoint
quadrilaterals that (non-trivially) touch one of its sides.
Since there are $c=8$ corners of $Q_A$ and $Q_B$, 
we have at most 16 quadrilaterals of the chain 
$Q_1, \dots, Q_k$ that are adjacent to corners of $Q_A$ and/or $Q_B$.

We now consider the quadrilaterals that are {\em not} adjacent to any
of these corners.

Let $Q_i$ be a quadrilateral that is not adjacent to any of the corners of the polygonal chains $A_1,A_2,A_3,A_4$ and $B_1,B_2,B_3,B_4$.
Two of its corners are adjacent to the same side $A_p$ of $Q_A$ 
and the 
other two are adjacent to the same side 
$B_q$ of $Q_B$, $1 \leq p,q \leq 4$. 
We call such a quadrilateral a \define{fair quadrilateral}.

\begin{lemma}
\label{lemma:pairOfFair}
For $k \geq 33$, 
in any touching quadrilaterals representation of $G^k$ 
there exists a pair of fair quadrilaterals $Q_i$ and $Q_{i+1}$ 
that are adjacent to the same sides of $Q_A$ and $Q_B$.
\end{lemma}


\begin{proof}
We can partition the set of fair quadrilaterals into 16 equivalence classes $C_{p,q}$, $1 \le p, q\le 4$, that denote the sets 
of fair quadrilaterals that are
adjacent to the same sides of $Q_A$ and $Q_B$. 
The equivalence class $C_{p,q}$ denotes that the pairs of sides $(A_p, B_q)$ are used.


Observe that if $Q_i$ is in an equivalence class $C$ and $Q_{i+1}$ is
not fair, then since $Q_{i+1}$ must be adjacent to a 
corner, $Q_{i+2}$ cannot be in the equivalence class $C$.
Thus, when we sweep through the chain of quadrilaterals $Q_1, \dots, Q_k$, 
we simultaneously proceed through the equivalence classes. 
By the pigeonhole principle, if there are at least $17$
fair quadrilaterals, then at least two of them must be in the same
equivalence class.
Combining that with the fact that there are at most $16$ 
quadrilaterals
that are not fair completes our proof.
\qed
\end{proof}

Before continuing with the proof of Theorem~\ref{thm:moreThan4}, 
we include the following observation, partially illustrated in Figure~\ref{fig-michael}(c):

{\bf Observation 2:}
If there are two regions $R,S$ touching in some nontrivial interval $I= (a,b)$ then at $a$, there is a corner
of $R$ or $S$. The same holds for corner $b$.

Using Observation~2, we see that each interval that is shared by two adjacent polygons ends at two of the corners of the two polygons.
Now, let $(Q_i, Q_{i+1})$ be a pair of fair same-sided quadrilaterals, touching sides $A_p$ and $B_q$. 
Since $Q_i$ is fair, the two corners associated with the adjacency of
$Q_A$ must belong to $Q_i$ and the other two corners of $Q_i$
are associated with the adjacency with $Q_B$.
The same applies for $Q_{i+1}$.
Since $Q_i$ and $Q_{i+1}$ have to be adjacent, the two sides next to each other touch. 
From Observation~2, at least two corners of $Q_i$
or $Q_{i+1}$ are involved in the adjacency.
For reference, label these two corners as $c_1$ and $c_2$.
The quadrilateral $Q_{l_i}$, corresponding to node $l_i$, must 
touch quadrilaterals $Q_A$, $Q_i$ and $Q_{i+1}$.
If $c_1$ (or $c_2$) were also associated with an adjacency
to $Q_A$ then $Q_{l_i}$ could not be adjacent to
all three quadrilaterals simultaneously.
Therefore, $c_1$ and $c_2$ must correspond to adjacencies with
$Q_B$.
A similar argument for $Q_{r_i}$ shows that
neither $c_1$ nor $c_2$ can correspond to adjacencies with $Q_B$
either.
However, this is a contradiction as all four corners
of both quadrilaterals are either associated with the adjacency with
$Q_A$ or with $Q_B$.
\qed
\end{proof}

\subsection{Five Sides Are Not Enough}

If we allow the regions to be pentagons, we must sharpen the argument a little more.

\begin{lemma}
\label{lemma:moreThanFive}
For $k \geq 71$, in any touching pentagons representation
for $G^k$, 
there exists a triple of fair pentagons $P_i, P_{i+1}, P_{i+2}$ 
adjacent to the same sides of $P_A$ and $P_B$.
\end{lemma}

\begin{proof}
We prove this along the same lines as the proof 
for Lemma~\ref{lemma:pairOfFair}.
As before, we can see that for a total of $10$ corners, at most 20 pentagons 
of the inner chain are not fair. 
The number of equivalence classes of pentagons with sides solely on the same side of $P_A$ and $P_B$ is at most $5\times 5=25$.
Recall that pentagons belonging to the same equivalence class
are sequential.
Since we aim now for triples and not just for pairs, using the
pigeonhole principle, if we have more than $2\times 25 = 50$ fair
pentagons at least three
must belong to the same equivalence class.
Therefore, as long as $k > 20+50=70$, there exists a triple
of fair same-sided pentagons.
\qed
\end{proof}

\begin{theorem}
For $k \geq 71$, there does not exist a touching polygons representation for $G^k$ in which all regions have complexity five or less.
\end{theorem}

\begin{proof}
From Lemma~\ref{lemma:moreThanFive}, 
let $(P_i, P_{i+1}, P_{i+2})$ be a triple of fair same-sided pentagons, touching sides $A_p$ and $B_q$. 
From Observation~2, we know that each interval that is shared by two polygons ends at two of the corners of the two polygons. 
Consequently, four of the five corners for $P_i$, $P_{i+1}$ and 
$P_{i+2}$ are adjacent to $A_p$ or $B_p$.
Since $P_i$ and $P_{i+1}$ have to be adjacent, the two sides next to each other touch. 
However, since there exist the polygonal regions representing $r_i$ and $l_i$, as before, the interval where $P_i$ and $P_{i+1}$ touch is disjoint from the regions $P_A$ and $P_B$.
As each region can have at most five corners, four of which
are adjacent to either $P_A$ or $P_B$, from Observation~2
we know that one corner from the adjacency with $P_i$ 
and $P_{i+1}$ belongs to $P_i$ and one belongs to $P_{i+1}$.
Similarly, we know that the adjacencies of
$r_{i+1}$ and $l_{i+1}$ imply that one corner of the adjacency
of $P_{i+1}$ and $P_{i+2}$ belongs to $P_{i+1}$ and the other belongs
to $P_{i+2}$.
Due to planarity, we also know that $P_i$ and $P_{i+2}$ lie
on opposite sides of $P_{i+1}$.
As these corners cannot be adjacent to $P_A$ or $P_B$,
we see that $P_{i+1}$ must have six distinct
corners, two adjacent to $P_A$,
two to $P_B$, one to $P_i$ and one to $P_{i+2}$, a contradiction.
\qed
\end{proof}

Note that six-sided polygons are indeed sufficient to represent the graph in Figure~\ref{fig-michael}(a). In particular, for fair polygons $P_i$ and $P_{i+1}$, we can use three segments on the lower side of $P_i$, while the upper side of $P_{i+1}$ consists of only one segment completely overlapping the middle of the three segments from the lower side of $P_i.$


\section{Touching Hexagons Representation}
\label{t6g}

In this section, we present a linear-time algorithm that takes as input a 
planar graph $G=(V,E)$ and produces a representation of $G$ in which all regions are convex hexagons.
This algorithm and the fact that every touching hexagons graph is
necessarily planar proves that the class of planar graphs
is equivalent to the class of touching hexagons graphs.

\subsection{Algorithm Overview}
\label{sec:t6gOverview}

We assume that the input graph $G=(V,E)$ is a fully triangulated planar graph with $|V|=n$ vertices. If the graph is planar but not fully triangulated, we can augment it to a fully triangulated graph with the help of dummy vertices and edges, run the algorithm below and remove the polygons that correspond to dummy vertices. 

Traditionally, planar graphs are augmented to fully triangulated graphs by adding edges to each non-triangular face. Were we to take this approach, however, when we remove the dummy edges we would have to perturb the resulting space partition to remove polygonal adjacencies. As this is difficult to do, we convert our input graph to a fully triangulated one by adding one additional vertex to each face and connecting it to all vertices in that face. The above approach works if the input graph is biconnected. Singly-connected graphs must first be augmented to biconnected graphs as follows. Consider any articulation vertex $v$, and let $u$ and $w$ be consecutive neighbors of $v$ in separate biconnected components. Add a new vertex $z$ and the edges $(z, u)$ and $(z,w)$. Iterating for every articulation point biconnects $G$ and results in an embedding in which each face is bounded by a simple cycle.
Since determining articulation points and 
adding vertices and edges to faces 
can be done in linear time,
the augmentation step incurs only a
linear amount of additional time to the main algorithm and adds
at most a linear number of vertices and edges to the original graph.


The algorithm has two main phases.
The first phase  computes the canonical labeling.
In the second phase we create regions with slopes 0, 1, -1 out of an initial isosceles right-angle triangle, by processing vertices in the canonical order. Each time a new vertex is processed, a new region is carved out of one or more already existing regions. At the end of the second phase of the algorithm we have a right-angle isosceles triangle that has been partitioned into exactly $n=|V|$ convex regions, each with at most 6 sides.  We show that creating and maintaining the regions requires linear time in the size of the input graph. 
We illustrate the algorithm with an example; see Figure~\ref{fig-sample}.

\begin{figure}[thb]
\begin{center}
\ifShowFigures
\begin{tabular}{cccc}
\includegraphics{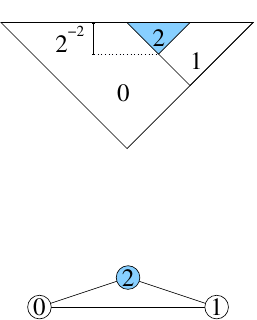} &
\includegraphics{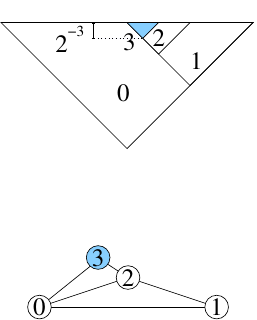} &
\includegraphics{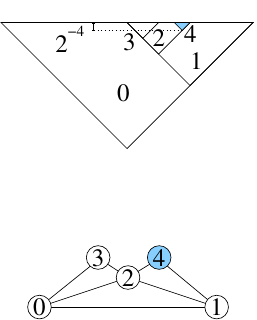} &
\includegraphics{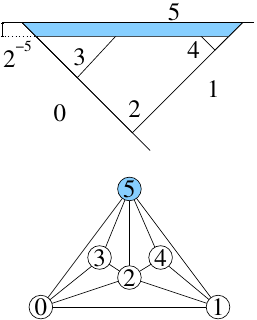}\\
(a) & (b) & (c) & (d)\\
\end{tabular}
\else
{\bf Figure skipped because causing problems with my printer!}
\fi
\caption{\small\sf \label{fig-sample} Incremental construction of the touching hexagons representation of a graph. Shaded vertices on the bottom row and shaded regions on the top row are processed at each step. In general, the region defined at step $i$ is carved at distance $2^{-i}$ from the active front on the top. Note that the top row forms a horizontal line at all times.}
\end{center}
\end{figure}


\vspace{-.4cm}\subsection{Region Creation\label{sec_t6galg}}
In this section we describe the $n$-step incremental process of inserting new regions in the order given by the canonical labeling, where $n=|V|$. The regions will be carved out of an initial triangle with coordinates $(0,0), (-1,1), (1,1)$. The process begins by the creation of $R_0$, $R_1$, and $R_2$, which correspond to the first three vertices, $v_0,v_1,v_2$; see Figure~\ref{fig-sample}(a). Note that the first three vertices in the canonical order form a triangular face in $G$ and hence must be represented as mutually touching regions. 

At step $i$ of this process, where $2 < i < n$, region $R_i$ will be carved out from the current set of regions. Define a region as ``active'' at step $i$ if it corresponds to a vertex that has not yet been connected to all its neighbors. An invariant of the algorithm is that all active regions are non-trivially tangent to the top side of the initial triangle, which we refer to as the ``active front.''

\iffalse
New vertices are created in one of two ways, depending on the degree of the current node, $v_i$, in the graph induced by the first $i$ vertices, $G_{i}$. By criterion 2 of the canonical labeling and the active regions invariant, $v_i$ is connected to 2 or more consecutive vertices on the outer face of $G_{i-1}$:

\begin{enumerate}
\item If $d_{G_i}(v_i)>2$ then $R_i$, the region corresponding to $v_i$, is a quadrilateral carved out of all but the leftmost and rightmost regions, by a horizontal line segment that is at distance $1/2^i$ from the active front; see Figure~\ref{fig-sample}(d). Note that all but the leftmost and rightmost neighbors of $v_i$ are removed from the set of active regions as their corresponding vertices have been connected to all their neighbors. Region $R_i$ is added to the new set of active regions. Call this a ``type 1 carving.''

\item If $d_{G_i}(v_i)=2$, let $R_a$ and $R_b$ be its neighbors on the frontier. Region $R_i$ is then carved out as a triangle from either $R_a$ or $R_b$ 
\end{enumerate}
\else
By criterion 2 of the canonical labeling and the active regions invariant, the current node $v_i$ 
is connected to two or more consecutive 
vertices on the outer face of $G_{i-1}$ and consecutive
regions on the active front.
Let $v_a$ and $v_b$ be the leftmost
and rightmost neighbors of $v_i$ on the outer face 
with corresponding (active front) faces $R_a$ and $R_b$.
The new region $R_i$ is defined to be an isosceles trapezoid 
formed by carving a horizontal line segment that is at
distance $1/2^i$ from the active front and intersects the
right side of $R_a$ and the left side of $R_b$.
The left (respectively, right) side of the trapezoid has
slope $-1$ (respectively, $+1$).
If the right side of $R_a$ has slope $+1$, a portion of its
region is necessarily carved out by $R_i$.
The same applies if the left side of $R_b$ has slope $-1$.
The regions between $R_a$ and $R_b$ have their upper segment
carved and no longer being tangential to the active front
are removed from the set of active regions.
In addition, $R_i$ is added to the list of active regions.
In Figure~\ref{fig-sample}(d), for example, both $R_0$
and $R_1$ have appropriate slopes and so are not carved
and $R_2$, $R_3$, and $R_4$ are all removed from the active front.

Note, that if $d_{G_i}(v_i) = 2$, then the length of the 
horizontal segment is $0$ and the shape is an isosceles triangle.
In this case, the geometry is such that exactly one
of $R_a$ or $R_b$ must necessarily be carved.
See Figures~\ref{fig-sample}(a-c).
\fi

\begin{lemma}
The above algorithm produces convex regions with at most 6 sides.
\end{lemma}

\begin{proof}
The convexity of the regions is obvious from the fact that regions are
created by a (linear) partitioning cut of a previous convex shape.
Note that the above algorithm leads to the creation of at most 
ten different types of regions; see Figure~\ref{fig-hierarchy}. 
Each region has a horizontal top segment, 
a horizontal bottom segment (possibly of length 0), 
and sides with slopes -1 or 1. 
Moreover, each region can be characterized as either opening 
(the first two in top row), static (the next four in the middle row), 
or closing (the last four in the bottom row), 
depending on the angles of the two sides connecting it to the 
top horizontal segment. 
At each iteration, each new region $R_i$ is an opening region.
In the new region's creation,
all affected regions except for $R_a$ and $R_b$ are carved with a 
horizontal line segment lying just below the top segment thus
having no effect on the shapes of these regions and removing
them from the active front.
Consequently, the only new region shapes possible stem from cutting $R_a$ and
$R_b$ when necessary with a slope $-1$ or $+1$ line respectively.
As the lower vertex formed by each cut is at least half the distance to 
the active front
from the previous vertices (those not on the active front), 
the only possible shapes are those shown in 
Figure~\ref{fig-hierarchy}.
Observe that closing regions cannot be carved at all.
\qed

\begin{figure}
  \begin{center}
\ifArxiv
    \includegraphics[width=0.8\textwidth]{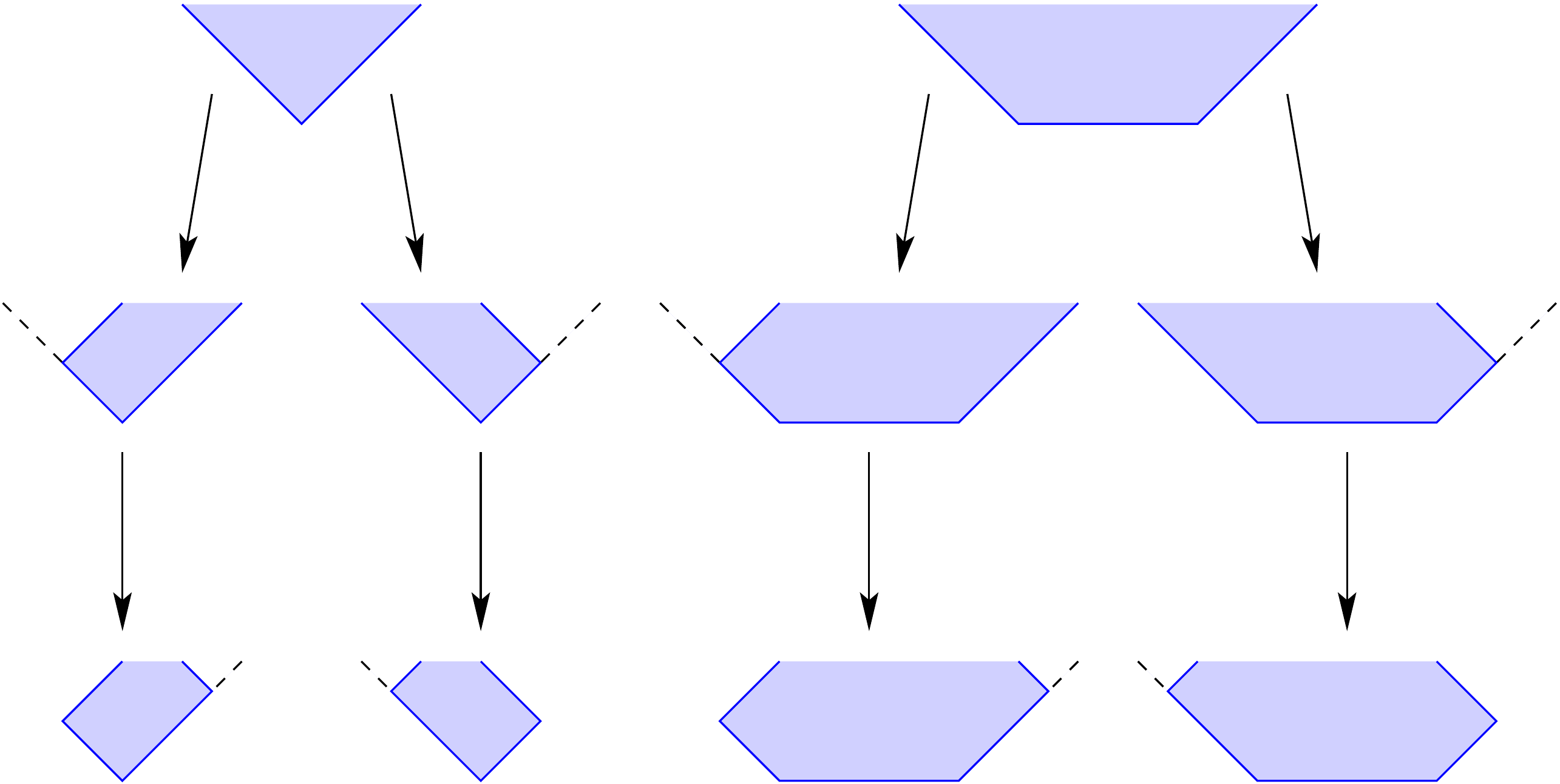}
\fi
  \end{center}
  \caption{\label{fig-hierarchy}
\small\sf There are a ten possible region shapes, falling into three categories: 2 opening, 4 static, and 4 closing. 
The arrows indicate carvings from one region to another.}
\end{figure}

\end{proof}

\vspace{-.4cm}\subsection{Running Time}
\ifArxiv
The above algorithm can be implemented in linear time. 
The linear-time algorithm for computing a canonical labeling of a planar graph~\cite{cp-ltadp-95} requires a planar embedding as an input. 
Recall that a planar embedding of a graph is simply a clockwise order of the neighbors of each vertex in the graph.  
Obtaining a planar embedding can be done in linear time using the algorithm by Hopcroft and Tarjan~\cite{ht-ept-74}. 

Creating and maintaining the regions in the second phase of our algorithm can also be done in linear time. We next prove this by showing that each region requires $O(1)$ time to create and requires $O(1)$ number of modifications.

Consider the creation of new regions. From criterion 2 of the canonical labeling, when we process the current vertex $v_i$, 
it is adjacent to at least two consecutive vertices on the outer face of $G_{i-1}$. By construction of our algorithm the vertices in the outer face of $G_{i-1}$ correspond to active regions and so have a common horizontal tangent. 

If $d_{G_i}(v_i)=2$, then a new region $R_i$ is carved out of one of the neighboring regions $R_a$ or $R_b$. 
Determining the coordinates of $R_i$ takes constant time, given the 
coordinates of $R_a$ and $R_b$ 
and the fact that $R_i$ will have height $1/2^i$ and 
will be tangent to the active frontier;
see Figure~\ref{fig-carving}(a).
When considering numerical precision, this exponential decrease in height
could lead to $O(n)$ bits needed per vertex; however, as we show in
Section~\ref{sec_compact} we only need the combinatorial structure and
not the precise geometric structure before
proceeding to draw the regions in a more compact and efficient form.

If $d_{G_i}(v_i)>2$, then all regions between and possibly including $R_a$
and $R_b$ will have their corresponding regions carved, 
in order to create the new region $R_i$;
see Figure~\ref{fig-carving}(b).
In this case the coordinates of $R_i$ can also be determined 
in constant time given the coordinates of $R_a$ and $R_b$ and 
the fact that $R_i$ will have height $1/2^i$ and 
will be tangent to the active frontier. 

\iffalse
\begin{figure}[!t]
\begin{center}
\includegraphics[width=14cm]{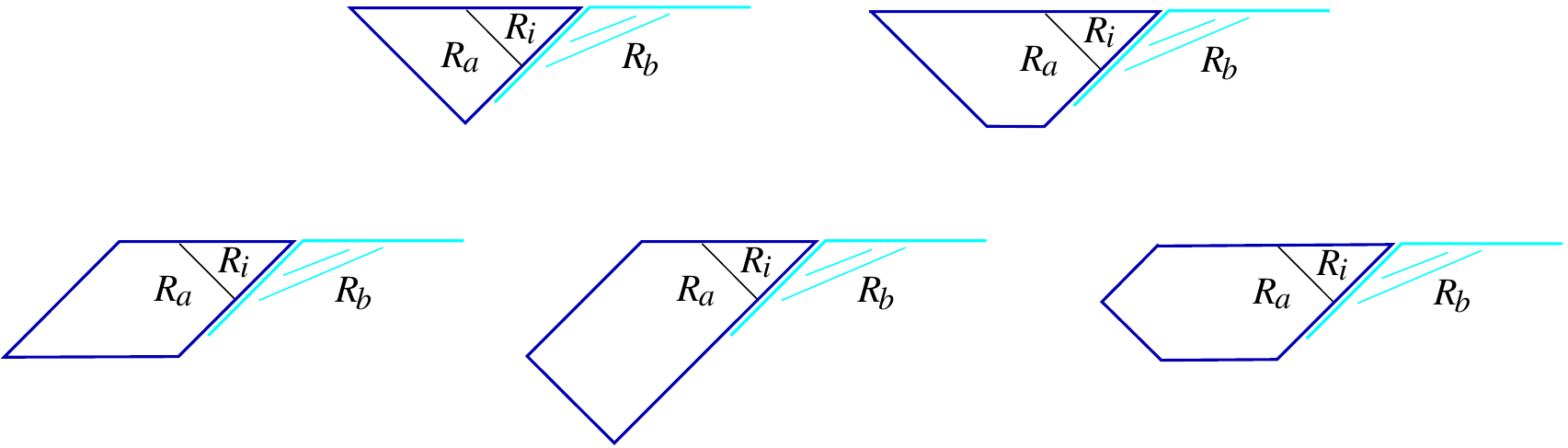}
\caption{\label{fig-carving}
\small\sf Introducing region $R_i$ between $R_a$ and $R_b$, assuming $R_i$ is carved out of $R_a$. All the possible cases are shown, assuming that $R_a$ and $R_b$ were convex, at most 6-sided regions with slopes 0, 1, -1. (There are five more symmetric cases when $R_i$ is carved out of $R_b$.) Note that these five regions correspond to the non-filled regions from the region-creating hierarchy in Fig.~\ref{fig-hierarchy} with two static regions in the first row and the three closing regions in the second row.}\vspace{-.4cm}
\end{center}
\end{figure}
\else
\begin{figure}
\begin{center}
\begin{tabular}{cc}
\input{exampleCutA.tex} & \input{exampleCutB.tex}\\
(a) & (b)
\end{tabular}
\caption{\label{fig-carving}
\small\sf Examples of introducing (the shaded) 
region $R_i$ between $R_a$ and $R_b$, 
(a) when $R_i$ is carved out of $R_a$ and $d_{G_i}(v_i) = 2$, and
(b) when $R_i$ is carved out of all neighbors of $v_i$ in $G_i$ except
$R_a$.}
\end{center}
\end{figure}
\fi

Consider the modifications of existing regions. 
As can be seen from the hierarchy of regions on Figure~\ref{fig-hierarchy},
there are exactly 10 different kinds of regions and each region begins 
as either an isosceles trapezoid or triangle 
and undergoes at most three modifications including the final horizontal
cut removing the region from the active front.
Moreover, once a region goes from one type to the next, 
it can never change back to the same type.
Finally note that the total number of region modifications is 
proportional to $|E|$ and since $G$ is planar, $|E|=O(|V|)$.
Thus, each region needs at most a constant number of modifications 
from the time it is created to the end of the algorithm. 

The algorithm described in this section, yields the following lemma:
\fi

\begin{lemma}
\label{lemma:t6g}
For any planar graph $G$ on $n$ vertices, we can construct
in linear time a touching hexagons representation
of $G$ with convex regions.
Moreover, if the graph is a triangulation, the representation is
also a tiling.
\end{lemma}

\section{Compaction Algorithm for Quadratic Area}
\label{sec_compact}

The algorithm given in Section~\ref{sec_t6galg} provides a touching hexagons representation of any planar graph. The incremental process carves out polygons within an ever smaller band of active front, therefore in practice the drawing is highly skewed, leading to exponential area. 
In this section we describe a compaction algorithm to get a 
drawing on an $O(n) \times O(n)$ grid.

When looking at the vertices and edges created in the algorithm for touching hexagons,
 if the horizontal edges are ignored, then the resulting graph is a ``binary'' tree, in the sense that each vertex has a degree of no more than 2. 
See Figure~\ref{fig:cappedTree}.
From this observation, we can generalize the compaction problem 
to the tree drawing routine described below.

We start with some definitions.
Order the nodes according to their inorder traversal.
A \define{cap set} is an ordered subset of the nodes such that
\begin{enumerate}
\item The first (resp. last) node has exactly one child, 
the left (resp. right) child.
\item All other nodes are leaf nodes,
with the addition that for the outermost cap
the first and last nodes are also leaf nodes.
\item The ordering of nodes in the cap set follows
the same inorder traversal ordering.
\item Any two cap sets are non-overlapping.
However, one may be nested in another,
in the sense that
if one set $C$ goes from node $a$ to node $b$ and is 
contained in a second set $C'$ then 
there exist two {\em consecutive} nodes $i,j$ in $C'$
such that $i < a < b < j$.
\end{enumerate}

\iffalse
\begin{figure}
\begin{center}
\includegraphics[width=\textwidth]{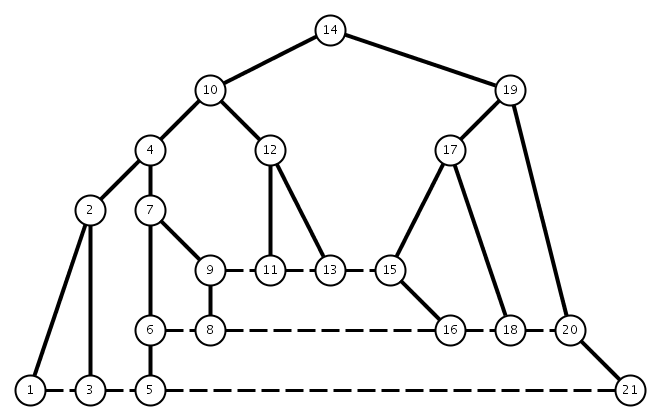}
\caption{\small\sf An example of a capped binary tree.
The nodes are numbered according to their inorder traversal.
The dashed lines represent the capped sets.
Note the nesting of all capped sets in this case is 
$(1,3,5, (6,8, (9, 11, 13, 15),
16, 18, 20 ), 21 )$.}
\label{fig:cappedTree}
\end{center}
\end{figure}
\else
\begin{figure}
\begin{tabular}{ccc}
\includegraphics[width=.25\textwidth]{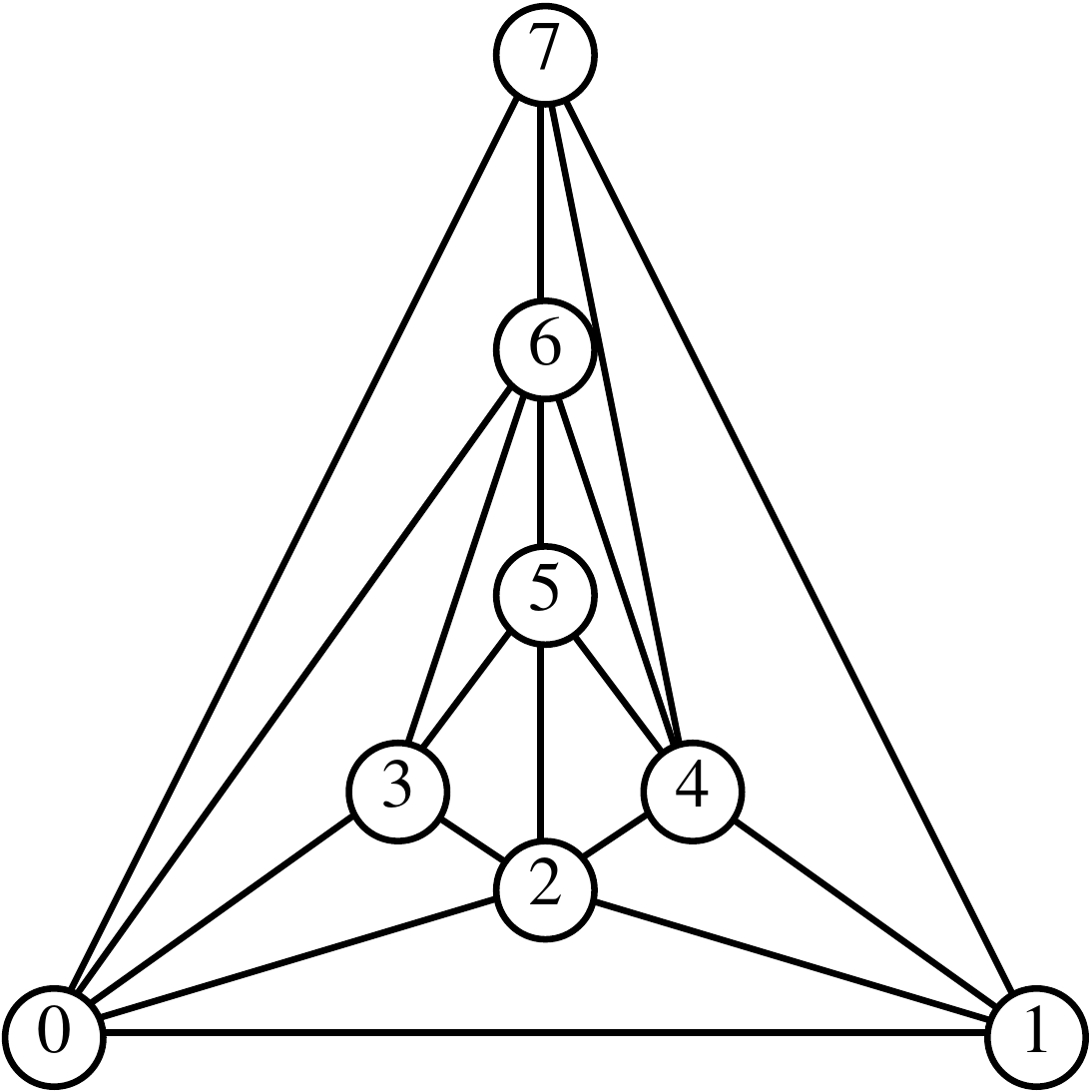} &
\hspace{-0.25in}
\includegraphics[width=.3\textwidth]{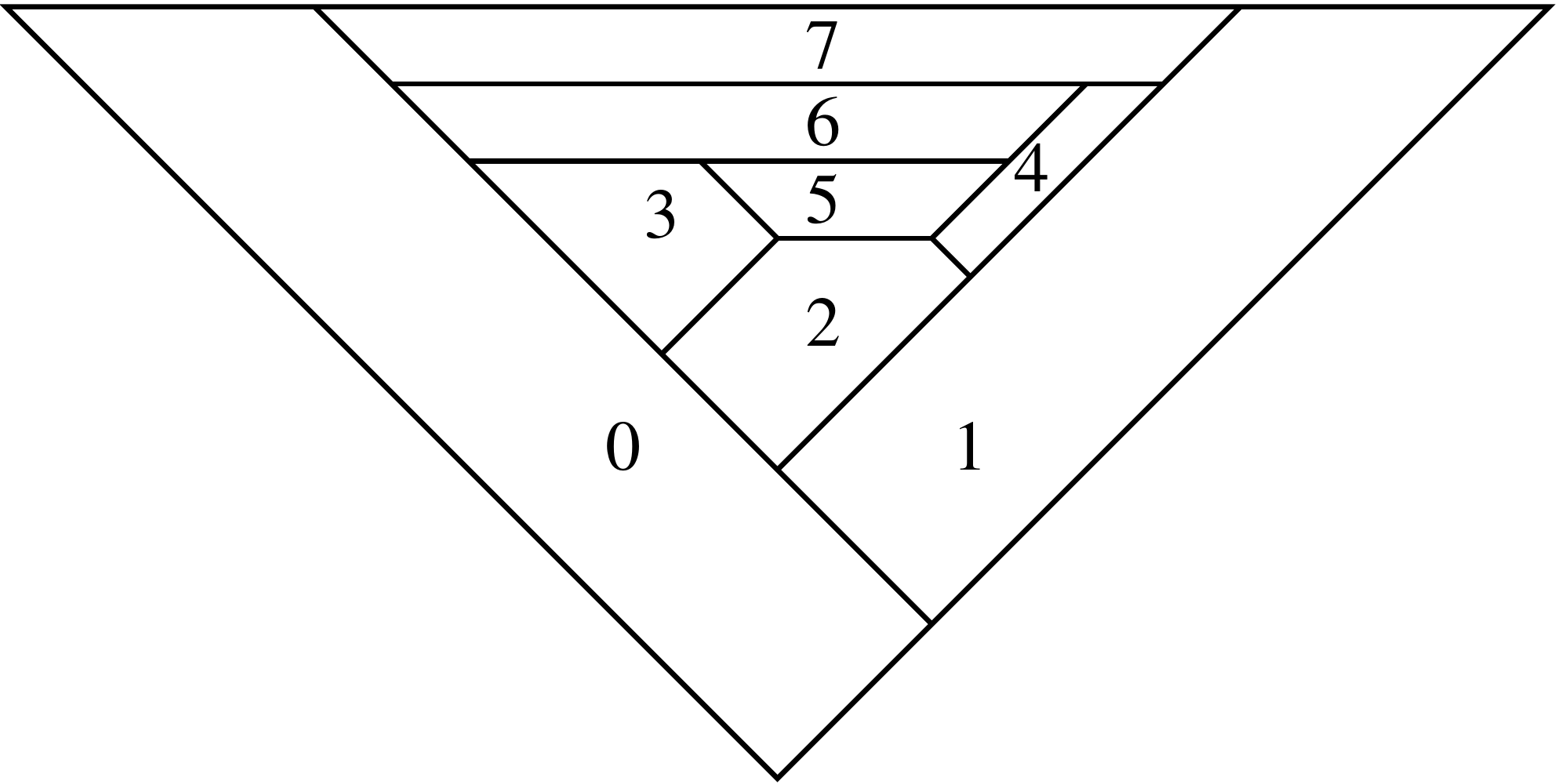} &
\hspace{-0.25in}
\includegraphics[width=.4\textwidth]{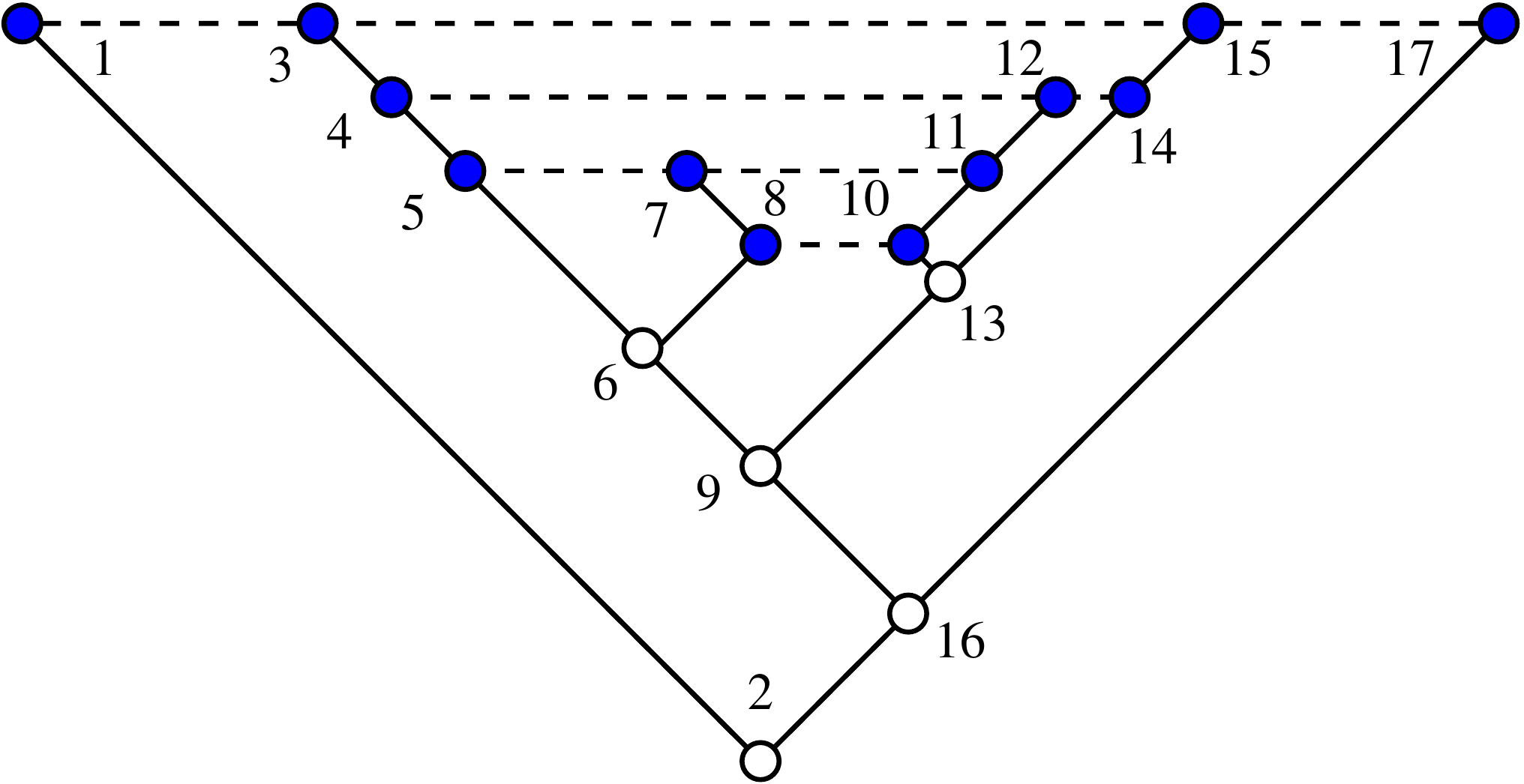}\\
(a) & 
\hspace{-0.25in}
(b) & 
\hspace{-0.25in}
(c)
\end{tabular}
\caption{\small\sf An example of a capped binary tree.
(a) The original graph.
(b) The hexagonal representation (not drawn to scale to conserve space).
(c) The corresponding capped binary tree with shaded nodes representing
cap nodes and dashed lines representing capped sets.
The nesting of all capped sets is 
$(1, 3, (4, (5, 7, (8, 10), 11), 12, 14), 15, 17)$.
We draw our binary tree upwards with the root at the bottom
to correlate better with the hexagonal drawing algorithm.}
\label{fig:cappedTree}
\end{figure}
\fi

Define a \define{capped binary tree} as a binary tree where every 
node either has two children (proper) or is assigned to a specific cap
set.
For convenience, we often refer to individual \define{cap nodes} in 
a cap set or to \define{cap node pairs} $(u,v)$ of neighboring (consecutive)
cap nodes in a cap set.

Figure~\ref{fig:cappedTree} illustrates the correlation
to the hexagons created by the 
algorithm in Section~\ref{sec_t6galg}
and provides an example of a capped binary tree,
where the nodes in the tree
represent the vertices formed in the drawing
(not the hexagon faces), the edges are precisely
the (non-horizontal) edges of the drawing, 
and each cap set is a maximal connected component of
vertices and horizontal edges of the drawing.

The \define{capped binary tree drawing problem} is to take
a capped binary tree and draw it on an integer grid such that
1) there are no edge crossings except at common endpoints;
2) each right (resp. left) edge is drawn with a slope of +1 (resp. -1); and
3) all nodes in a cap set are drawn with the same $y$-coordinate
such that they can be connected by a horizontal line segment
without crossing any other edges (except at the nodes in the set).

Before proving that we can draw capped binary trees on an $n/2 \times n/2$
grid, where $n$ is the number of nodes in the tree, we first 
present a divide-and-conquer compaction algorithm to accomplish this.
The algorithm is inspired by the layered tree drawing 
algorithm~\cite{Walker_1990}
with the additional aforementioned constraints.

Let $G_{T6G}$ be the graph derived from the touching hexagons algorithm,
{formed by taking the vertices as the
intersections of the regions and the edges as (portions of)
the sides of the regions connecting the vertices;
see Figure~\ref{fig:cappedTree}(c) where
the edges are both the solid and dashed lines.}
Let $G_T$ be the corresponding capped binary tree, formed
by removing the horizontal edges.
Our compaction algorithm, described in detail below, 
proceeds by incrementally removing
and placing a subset of the leaf nodes 
from a subtree $G_c$ initially set
to $G_T$.
For each of the leaf nodes of $G_c$ removed,
the resulting placement requires the node to be connected
to all of its child nodes (if any) in the original tree $G_T$
and also
might require adjusting the position of one of its subtrees.
Our process works by only adjusting the horizontal positions
of any node, thus preserving a subtree's vertical position.
For performance reasons, we actually delay the horizontal
shifting by merely recording the shift needed for a subtree
in its root.
The shifts are then propagated through the tree in a final 
post-processing stage.
Initially, the $x$ and $y$ positions as well as the
horizontal shift of every node is set to $0$.

Before proceeding with the details of the algorithm, we
clarify precisely those leaf nodes that are removed and
placed at each iteration.
Define the \define{active front node set} $F$ of $G_c$
as the maximum subset of leaf nodes of $G_c$, 
such that a cap node is in 
$F$ if and only if all the nodes in its cap set are also in $F$.
The initial active front is precisely those vertices at the 
upper edge of the outer triangular 
\ifArxiv
region; see Figure~\ref{fig-compExample-I}.
\fi

\begin{enumerate}[1.]
\item For each node $v$ in the active front node set $F$ of $G_c$,
\begin{enumerate}[a)]
\item if $v$ is a leaf in $G_T$, we do nothing ($v$ remains at $(0,0)$).
\item if $v$ has one subtree in $G_T$ and if it is to the right,
extend a slope +1 line from the root of this right subtree by 
1 unit down and left to get the position of $v$.
If it instead has a left child, extend a slope -1 line, down and right.
\item 
\label{enum:stepSeparation}
if $v$ has two subtrees in $G_T$, shift the right subtree 
horizontally so that the two subtrees have 
a ``separation''\footnote{We elaborate on 
what separation entails shortly.}
of 
either distance 1 or 2, and the slope -1 (resp. +1) line from 
the root of the left (resp. right) subtree 
meet at a grid point, the
assigned point for $v$.
Record the shift used at the root of the right tree.
\end{enumerate}
\item 
\label{enum:stepHeight}
For each cap set $C$ in the front, 
set $h$ to be the maximum of the absolute values of $y$ coordinates of the cap nodes in $C$. For every cap node $v \in C$,
\begin{enumerate}[a)]
\item if $v$ is a leaf node in $G_T$, set $y(v) = -h$;
\item otherwise, by construction, node $v$ must have only one 
subtree.
If it is to the right, 
extend the slope +1 line from the root of the subtree 
till it intersects with the line $y = -h$, 
and record the coordinates of the intersection point 
as the coordinates for $v$. 
If it is to the left, extend the slope -1 line instead.
\end{enumerate}
\item Delete $F$ and its connecting edges from $G_C$, renaming the resulting tree $G_C$. If $G_c$ is not empty, go to Step 1.
\item Propagate the horizontal shifts from each node to its subtree via
a pre-order traversal starting at the root of $G_T$ to obtain a final 
integer grid position for each node.
\end{enumerate}

This algorithm yields a drawing of the $G_{T6G}$ on a grid. 
\ifArxiv
Figures~\ref{fig-compExample-I}--\ref{fig-compExample-IV} 
illustrate execution of this algorithm on the graph
from Figure~\ref{fig:cappedTree}.
Figure~\ref{fig-samples} shows some additional graphs, 
and their corresponding touching hexagons representations on a grid.
\fi
Because the algorithm processes entire cap sets at a time,
because Step~\ref{enum:stepHeight} places all nodes in the same cap set
at the same (lowest) height, and because it only shifts nodes horizontally,
all cap nodes are drawn at the same vertical position.
Further, because
 the algorithm also only connects the tree edges using line segments
with slope $\pm 1$ and applies any horizontal shifts to the entire
subtree via the final propagation step, all tree edges
are drawn with slopes $\pm 1$.
Consequently, the drawing produced by this algorithm is a valid
capped binary tree drawing.

We also need to show that the grid size used is reasonable.
To bound the area, we must first elaborate on the 
compaction step (Step~\ref{enum:stepSeparation})
that combines
two trees such that their separation is either distance 1 or 2.
This separation is not between the two roots of the subtrees but
between the closest two nodes.
In essence, we wish to compact the two subtrees as close as possible.
\ifArxiv
Although there are several possible approaches including
some that are more straightforward but slower than ours, 
we describe below a simple method to ensure that the algorithm's
time remains linear.

For any subtree $T$, if we examine the cap nodes we see
that some cap nodes are complete in the sense that all of
the nodes in their cap set belong to $T$.
We are interested in the cap sets that are not complete.
A cap node is a \define{left cap node} of $T$ if it is 
the smallest (numbered) node in its cap set that is also in $T$ but
not the smallest node in its entire cap set; essentially, the left side
is not complete.
We define a \define{right cap node} similarly.
See Figure~\ref{fig:zipExample}(a).
Note that it is possible that a node is both a left and right cap
node.
Because of the fact that all non-cap nodes have two children
and all left (respectively, right) cap nodes have only right
(respectively, left) children if any, when joining two
subtrees the two closest nodes will necessarily be between
a left and right cap node (of the same cap set).
We therefore maintain the list of left and right cap nodes
and their respective offsets from the root.
\fi
\begin{figure}
\begin{center}
\begin{tabular}{cc}
\ifArxiv
\includegraphics[scale=.3]{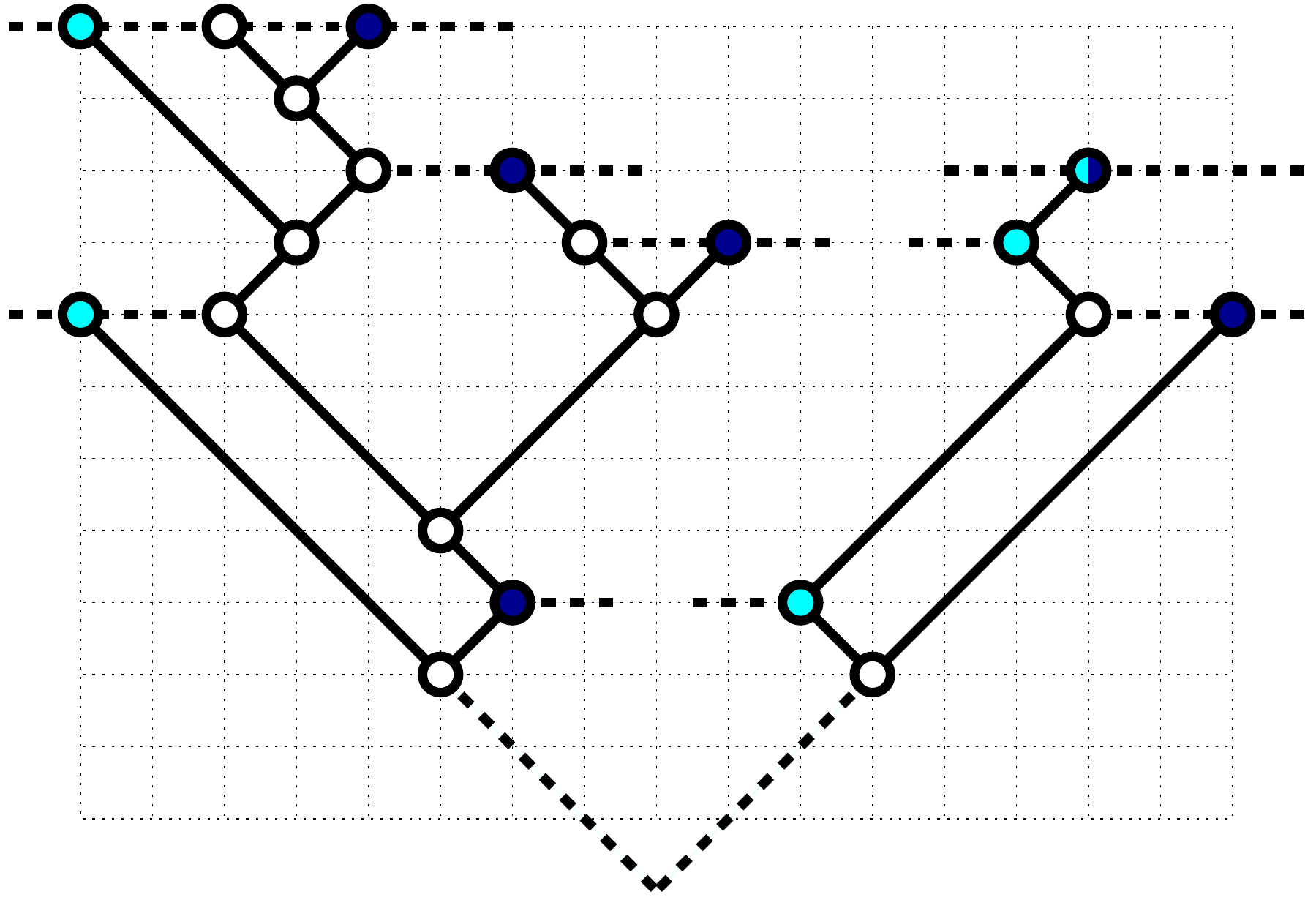} &
\includegraphics[scale=.3]{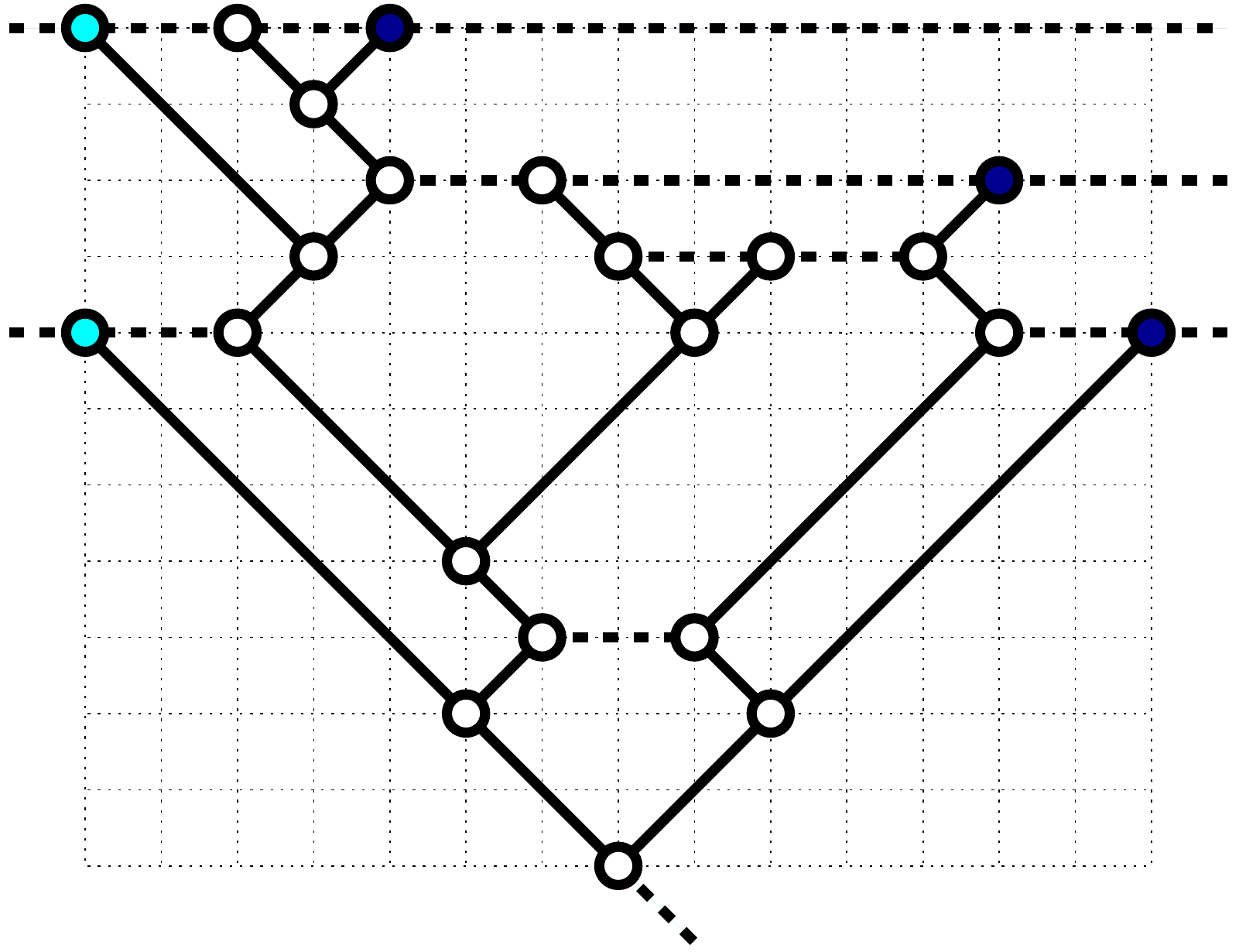} \\
\fi
(a) & (b)\\
\end{tabular}
\end{center}
\ifArxiv
\caption{\small\sf An example of compacting two subtrees together during
Step~\ref{enum:stepSeparation}.
(a) The two subtrees with left and right cap nodes shaded 
light and dark respectively.
Observe that the top node in the right subtree is two-tone
to reflect that it is both a left {\em and} right cap node.
(b) The resulting tree after merging with the updated left and right
cap nodes highlighted.}
\fi
\label{fig:zipExample}
\end{figure}

\ifArxiv
Let $T(C_L)$ represent the set of left cap nodes
and $T(C_R)$ represent the right cap nodes of subtree $T$.
We maintain both sets as a doubly-connected circularly linked list
ordered by their $y$-values.
Each cap node maintains its relative $x$-distance from the node just
before and after it in this list.
In addition, we know the offset of the first cap node from the root.
When merging two subtrees $T_1$ and $T_2$
where $T_1$ lies to the left of $T_2$,
we can find the closest distance between
them by starting with the 
first (right) cap node $r$ in the list $T_1(C_R)$, the one closest to the root 
and the first (left) cap node $l$ in the list $T_2(C_L)$.
Observe that $r$ and $l$ must be in the same cap set.
We then proceed to zip up the lists until the first one finishes,
without loss of generality assume it is $T_2$.
After each step, we can compute the offset of each cap node 
relative to its corresponding root and thus determine the
closest we can bring the two subtrees.
We can also maintain the list for the merged tree $T$ by setting
$T(C_L) = T_1(C_L)$ and $T(C_R) = T_2(C_R)$ and then 
merging the remaining elements from $T_1(C_R)$ into $T(C_R)$.
The same applies if $T_1$ finishes first (has the shorter list).
Figure~\ref{fig:zipExample}(b) shows the resulting merge of two
subtrees.
\fi

\begin{lemma}
\label{lemma:capped}
Given any capped binary tree $T$, we can compute in linear time
a capped binary tree drawing of $T$ on an $(n-1) \times (n-1)/2$ grid.
\end{lemma}

\begin{proof}
\ifArxiv
We have already described the linear-time algorithm that
produces a valid compact drawing on an integer grid.
\fi
However, it still remains to prove that the resulting drawing
is sufficiently compact.
We do this by inductively analyzing the separation between neighboring 
cap nodes, which are ``joined'' during the process described above.
This proof is reminiscent of the one given by 
Kant~\cite{kant-92}.

For every cap node pair $(u,v)$, consecutive nodes in a cap set, 
let its \define{interior cap set} $C_i$ be the cap
set (if one exists) whose first node is the next cap node 
in the inorder traversal of $T$ from $u$.
Note that from the definition of the nesting of cap sets, 
the last node
in $C_i$ would be the last cap node before $v$.
If no such set exists, then let $C_i$ be $\{\mbox{lca}(u,v)\}$,
where  $\mbox{lca(u,v)}$  represents the least common ancestor of
$u$ and $v$.
If $u$ has a child node, it must be a left child and if
$v$ has a child node, it must be a right child.
Since any subtree necessarily has at least one cap node
(a leaf of that subtree),
this lca is also the only node in the inorder
traversal between $u$ and $v$.
We refer to $u$ and $v$ as the \define{exterior cap nodes} of $C_i$.
Observe that the interior cap set will have its $y$-coordinate value 
closer to the root in the final drawing.
For example, in Figure~\ref{fig:cappedTree}(c), the cap node pair $(4,12)$
has the interior cap set $\{5,7,11\}$, whereas the cap node pair $(12,14)$
has as the interior cap set the lca $\{13\}$.

Our proof uses the following inductive claim.
After every iteration of our algorithm, for any subtree $T'$ and any cap node pair $(u,v)$ 
in $T'$ that is part of a cap set whose exterior cap nodes are not also in $T'$
or is the outermost cap set,
the (horizontal) distance between $(u,v)$ is no more than twice the 
number of cap node pairs
in the inorder traversal of $T'$, inclusive of the pair $(u,v)$.

If this claim holds, then the drawing is on an $(n-1) \times (n-1)/2$ 
grid because the
final drawing is a single tree with one row of cap nodes at the top 
whose width cannot exceed twice the number of cap node pairs in the tree.
To bound the number of cap node pairs, notice that the graph formed
by the binary tree edges combined with the horizontal (capping) edges forms a
3-regular graph, excluding the root and leftmost and rightmost vertices
which have degree two.
Since this graph has $(3n-3)/2$ edges and $n-1$ of the edges are tree edges,
that leaves exactly $(n-1)/2$ horizontal edges.
Each horizontal edge corresponds to a unique cap node pair.
The height follows from the $\pm 1$ slope of the non-horizontal edges.

Initially, every node is in its own subtree so the claim holds.
Inspecting the algorithm reveals that the only place where the claim could change is
in Step~\ref{enum:stepSeparation} where two subtrees $T_1$ and $T_2$ are merged.
In addition, since the trees are simply shifted to merge, the only possible change
is due to the introduction of new cap node pairs, a cap node from each subtree is aligned
with its neighbor in the other tree.
In fact, since the merging process zips nested cap sets in succession,
we are only concerned with the final width of the last cap pair merged.
We again prove this by induction on the zipping process.

We claim that the width of the cap node pair $(u,v)$ merged is no more than
twice the number of cap node pairs in its inorder traversal from $u$ to $v$.
Let $(u,v)$ be the first cap node pair merged.
Since the interior cap set of $(u,v)$ is simply $r$, which is the only
node in the inorder traversal between $u$ and $v$,
the resulting width at this stage is at most 2; 
see Figure~\ref{fig:mergeTree}(a).
Thus, our claim holds after the first merged cap node pair.

We now progress inductively.
Let $(u,v)$ be the next cap node pair merged with width $\ell$,
$(u',v')$ be the previous pair, and
$C$ be the interior cap set of $(u,v)$.
Notice that $(u',v') \in C$.
By induction, we know that the entire width $\ell'$ of $C$ is no more than 
twice the 
number of cap node pairs 
in the inorder traversal from the first to last cap nodes in $C$.
In addition, since $u$ has no right subtree and $v$ has no left subtree,
the next cap node in the inorder traversal is the first cap node in $C$
and the last in the traversal is the last cap node in $C$.
Therefore, 
we know that the number of cap node pairs in the inorder traversal 
from $u$ to $v$ is the same as the number for $C$ plus one, the one 
for $(u,v)$.
Therefore, we need only to prove that $\ell \leq \ell' + 2$.

\begin{figure}
\begin{center}
\begin{tabular}{ccc}
\input{mergeLCAExample.tex} &
\input{mergeExample.tex} &
\includegraphics[width=.25\textwidth]{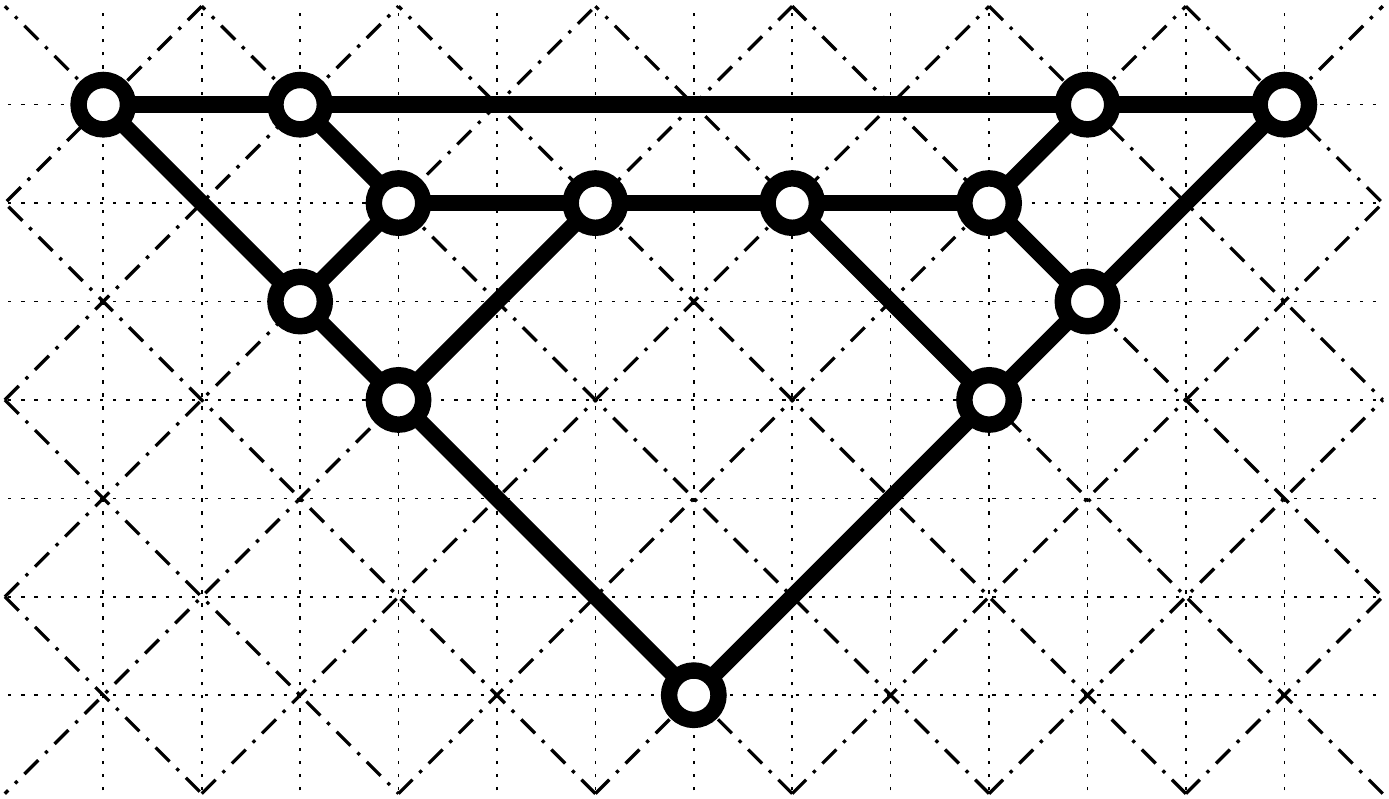}\\
(a) & (b) & (c)
\end{tabular}
\caption{\small\sf Example of the merging of two subtrees during successive cap
node pairs for (a) a single lca node and $(u,v)$ and
(b) $(u',v')$ and $(u,v)$.
(c) A simple drawing of a capped binary tree highlighting both 
the normal grid (horizontal/vertical lines) 
and the rotated, space-efficient, grid (diagonal lines).
\label{fig:mergeTree}}
\end{center}
\end{figure}

Let $a$ and $b$ be the first and last cap node in $C$; see 
Figure~\ref{fig:mergeTree}(b).
By the definition of a capped binary tree, we know that $a$ and $b$ each 
have only one child, a left and a right respectively.
In addition, by Step~\ref{enum:stepHeight} we know that one of the two child nodes
is only one unit above its parent.
Without loss of generality assume it is the left child of $a$.
This means that the node is also one unit to the left of $a$.
Node $u$ is a descendant of this left child but from the definition
of the capped binary tree, $u$ can be found by traversing successive right children only.
Therefore, the path from this left child to $u$ follows a straight line of slope $+1$.
This follows parallel with the line from $b$ through its right child.
The (horizontal) distance from $u$ to this line is exactly $\ell'+2$.
Since the path from this right child to $v$ follows left children only (if any),
the distance from $u$ to $v$ is $\ell \leq \ell'+2$.
This completes our proof.
\qed
\end{proof}

For a clearer understanding and better symmetry 
with the construction technique used in Section~\ref{t6g}, we used
edges with slopes $\pm 1$ and $0$.
\ifArxiv
We can improve the area bound slightly using a rotated drawing
yielding the following corollary:
\fi
\begin{corollary}
\label{cor:capped}
Given any capped binary tree $T$, we can compute in linear time
a (rotated) capped binary tree drawing of $T$ on an $(n-1)/2 \times (n-1)/2$ 
grid.
\end{corollary}

\ifArxiv
\begin{proof}
The trick is to rotate the grid $45^\circ$.
The result is that the tree edges are drawn as horizontal and vertical lines
while the ``previously horizontal'' capped edges are drawn with slope $-1$.
As Figure~\ref{fig:mergeTree}(c) illustrates the only places where the 
vertices can lie from the initial drawing are the same as the overlaid 
rotated grid of dimension $(n-1)/2 \times (n-1)/2$.
This can be further shown by observing that every vertex is connected to 
the root through a sequence of binary tree edges (of slope $\pm 1$ 
in the original grid).  Thus, the original grid points lying
in between the rotated grid squares cannot contain any vertices
from the drawing.
Thus, we have the same drawing on a slightly more compact grid.
\qed
\end{proof}
\fi

Since the initial construction step does not
need to create the hexagons explicitly,
that step can be used simply to determine
the combinatorial representation of the capped binary tree.
This prevents any issues with numerical precision and representation.
Combining Lemmas~\ref{lemma:t6g} and~\ref{lemma:capped} yields our first proof
for Theorem~\ref{thm:t6g}.


\section{Another Hexagonal Representation using $O(n)\times O(n)$ Area}
\label{sec_alternate}

In this section, we present an alternative approach
to proving Theorem~\ref{thm:t6g}.
This approach is based on Kant's algorithm for hexagonal grid drawings of 3-connected, 3-regular planar graphs~\cite{kant-92}.
Although the modification needed is direct, we feel that our previous
approach is
a more intuitive and constructive technique that yields better fundamental
insight into the nature of the problem.

In Kant's algorithm the drawing is obtained by looking at the dual graph and processing its vertices in the canonical order. In the final drawing, however, there are two non-convex faces, separated by an edge not drawn as a straight-line segment. We address these problems by adding some extra vertices in a pre-processing step. Once the dual of this augmented graph is embedded, the faces corresponding to the extra vertices can be removed to yield the desired $O(n)\times O(n)$ grid drawing. 

Let $H=(V,E)$ be a 3-connected, 3-regular planar graph. 
Note that the dual $D(H)$ is fully triangulated, as each face in the dual
corresponds to exactly one vertex in $H$. So, for $f$ faces in $H$, 
we have $f$ vertices in $D(H)$. 
We first compute a canonical ordering on the vertices of $D(H)$ as defined
by de Fraysseix {\em et al.}~\cite{FMR04}.
Let $v_1, \dots, v_f$ be the vertices in $D(H)$ in this canonical order.

Kant's algorithm now constructs a drawing for $H$ on the hexagonal grid
such that all edges but 
one have slopes $0^\circ$, $60^\circ$ or $-60^\circ$,
with the one edge with bends lying on the outer face. 
The typical structure of those drawings is shown in Figure~\ref{fig:kant}(a).
Although we focus our description using the hexagonal grid, 
to place on the rectilinear grid, the corresponding slopes are
$0^\circ$, $90^\circ$ and $-45^\circ$.


\begin{figure}[tb]
\begin{center}
\ifShowFigures
\begin{tabular}{cc}
\ifArxiv
\includegraphics[width=.3\textwidth]{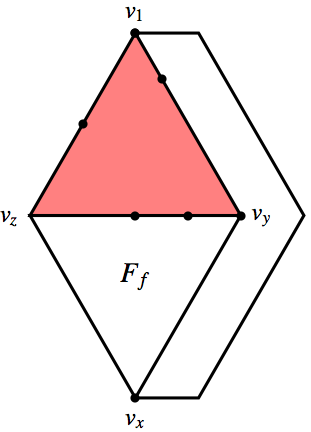} &
\includegraphics[width=.4\textwidth]{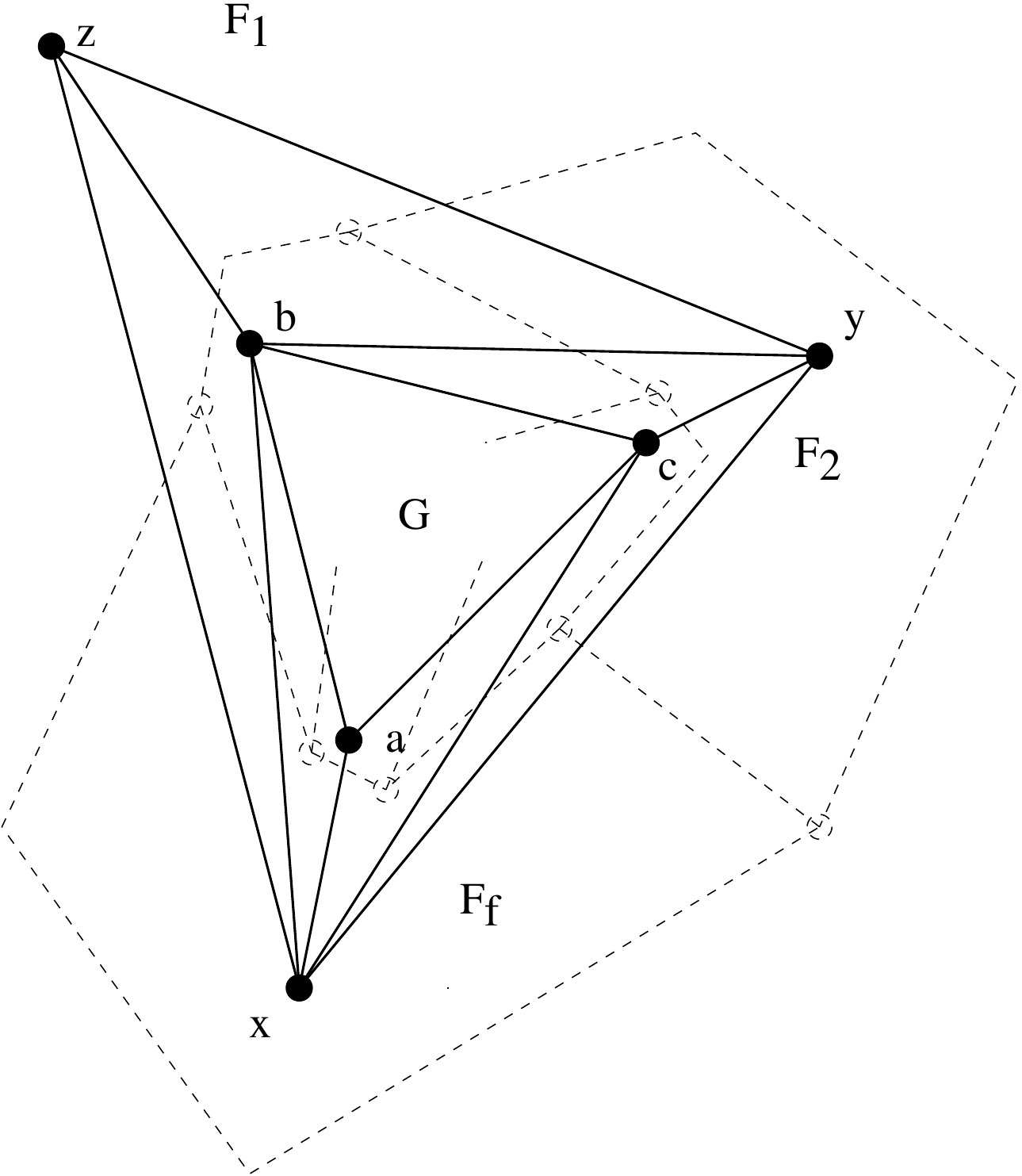}\\
\else
\includegraphics[width=.2\textwidth]{kant.png} &
\includegraphics[width=.3\textwidth]{Gplus3}\\
\fi
(a) & (b)
\end{tabular}
\else
{\bf This figure skipped as well!!}
\fi
\caption{\small\sf (a) Polygonal structure obtained from Kant's algorithm. 
(b) Graph $G$ augmented by vertices $x$, $y$ and $z$ 
together with its dual which serves as the input graph 
for Kant's algorithm.\label{fig:kant}}
\end{center}
\end{figure}

The algorithm incrementally constructs the drawing by adding the faces of $H$ 
in reverse order of the canonical order of the corresponding vertices in $D(H)$.
We let $w_i$ be the vertices of $H$.
Let face $F_i$ correspond to vertex $v_i$ in $D(H)$. The algorithm starts with a
triangular region for the face $F_f$ that corresponds to vertex $v_f$.
The vertex $w_x$ that is adjacent to $F_f$, $F_1$ and $F_2$ is placed at the
bottom.
Let $w_y$ and $w_z$ be the neighbors of $w_x$ in $F_f$. These three vertices 
form the corners of the first face $F_f$. 
$(w_x,w_z)$ and $(w_x,w_y)$ are drawn upward with equal lengths and slopes 
$-\sqrt{3}$ and $\sqrt{3}$, respectively.
All the edges on the path between $w_y$ and $w_z$ along $F_f$
are drawn horizontally between the two vertices. 
From this first triangle, all other faces are added in reverse canonical 
order to the upper boundary of the drawing region. 
If a face is completed by only one vertex $w_i$, this vertex is placed 
appropriately above the upper boundary such that it can be connected by 
two edges with slopes $-\sqrt{3}$ and $\sqrt{3}$, respectively. 
If the face is completed by a path, then the two end segments of the path 
have slopes $-\sqrt{3}$ and $\sqrt{3}$, while the other edges are horizontal.
The construction ends when $w_1$ is inserted, corresponding to the outer
face $F_1$. Note that there is an edge between $w_1$ and $w_x$, which is drawn
using some bends. 
This edge is adjacent to the faces $F_1$ (the outer face) and $F_2$.

From this construction, we can observe that the angles at faces 
$F_f, \dots, F_3$ have size $\leq 180^\circ$ as the first two edges
do not enter the vertex from above, and the last edge leaves the vertex upwards.
Hence, we have the following result. 

\begin{lemma}
The faces $F_f,\dots,F_3$ are convex, and as the slopes of the 
edges are $\pm\!\sqrt{3}$ or $0$,
they are drawn with at most six sides.
\end{lemma}

This property is exactly what we are aiming for, as the vertices of our 
input graph $G$ should be represented by 
convex regions of at most six sides.
Unfortunately, Kant's algorithm creates two non-convex faces $F_1$ and $F_2$
separated by an edge which is not drawn as a line segment.
Furthermore, the face $F_f$ is drawn as large as all the remaining faces $F_3,\dots,F_{f-1}$ together.

Kant gave an area estimate for the result of his algorithm
which is the same for both hexagonal and rectilinear grids.
A corollary of Kant's algorithm is the following:

\begin{corollary}
For a given 3-connected, 3-regular planar graph $H$ of $n$ vertices, 
$H - w_x$ can be drawn
within an area of $n/2 -1 \times n/2 -1$.
\end{corollary}

To apply Kant's result to the problem of constructing a touching hexagons representation, we enlarge the embedded input graph 
$G$ so that the dual of the resulting graph $G'$ can be drawn
using Kant's algorithm in such a way that the original 
vertices of $G$ correspond to the faces $F_3,\dots,F_{f-1}$.

We add 3 vertices corresponding to faces $F_1, F_2$ and $F_f$ in Kant's algorithm.
Since $G$ is fully triangulated, let $a, b$ and $c$ be the vertices at the outer face of $G$ in clockwise order.
We add the vertices $x$, $y$ and $z$ in the outer face and connect to $G$
so that $z$ corresponds to the outer face $F_1$, $y$ to $F_2$, and 
$x$ to $F_f$.
First, we add $x$ and connect it to $a$, $b$ and $c$ such that $b$ 
and $c$ are still in the outer face.
Then we add $y$ and connect it to $x$, $b$ and $c$ such that $b$ is 
still in the outer face.
Finally, we add $z$ and connect it to $x$, $y$ and $b$ such that $x$, 
$y$ and $z$ now form the outer face;
see Figure~\ref{fig:kant}(b).

Since the vertices $x$, $y$ and $z$ 
are on the outer face, we can choose which one 
is first, second and last in the canonical order. We then apply
Kant's algorithm with the canonical order $v_1 = z, v_2 = y$ and $v_f = x$.
After construction, we remove the regions corresponding
to vertices $x$, $y$ and $z$, yielding a hexagonal representation
of $G$.

Given any (connected) planar graph $G$, we can make it fully triangulated
using the technique described in Section~\ref{sec:t6gOverview}.
We can then remove the added vertices and edges.
Since Kant's algorithm runs in linear time, and our emendations can be
done in linear time, we get another proof for Theorem~\ref{thm:t6g}.
We again use at most three slopes for each representation with
sides having slopes $\pm\!\sqrt{3}$ or $0$
(or $0$,$+\infty$ and $-1$).

For a triangulated input graph $G=(V,E)$, 
we have $n$ vertices and, by Euler's formula, $2n-4$ faces.
Since we enhanced our graph to
$n+3$ vertices, we have $f = 2n+2$ faces. Those faces are the vertices in the
dual $D(G)$ which is the input to Kant's algorithm.
His area estimation gives an area of $f/2 -1 \times f/2 -1$ for $f$ 
vertices when we coalesce 
the faces $F_1, F_2$ and $F_f$ into a single outer face by 
removing the corresponding vertices and edges. 
Thus, we get an area bound of $n \times n$
using exactly the same argument as in~\cite{kant-92}.

\section{Conclusion and Future Work}

Thomassen~\cite{Thomassen} had shown that not all planar graphs can be represented by touching pentagons, where the external boundary of the figure is also a pentagon and there are no holes. 
Our results are more general, as we do not insist on the external boundary being a pentagon or on there being no holes between pentagons. It is possible to derive algorithms for convex hexagonal representations for general planar graphs from several earlier papers, e.g., de Fraysseix {\em et al.}~\cite{FMR04}, Thomassen~\cite{Thomassen}, and Kant~\cite{kant-92}. 
However, these
do not immediately lead to algorithmic solutions to the problem of 
computing a graph representation with convex low-complexity 
touching polygons. To the best of our knowledge, this problem has never been formally considered.

In this paper, we presented several results about touching $k$-sided graphs. 
We showed that, for general planar graphs, six sides are necessary and sufficient,
and that the algorithm for creating a touching hexagons representation can be modified
to yield an $O(n) \times O(n)$ drawing area.
Finally, we discussed a different algorithm for general planar graphs 
which yields a similar drawing area.

Several interesting related problems are open. What is the complexity of deciding whether a given planar graph can be represented by touching triangles, quadrilaterals, or pentagons? 
In the context of rectilinear cartograms, the vertex-weighted problem has been carefully studied. 
However, the same problem without the rectilinear constraint has received less attention. Finally, it would be interesting to characterize the subclasses of planar graphs that allow for touching triangles, touching quadrilaterals, and touching pentagons representations.

\begin{acknowledgements}
We would like to thank Therese Biedl for pointing out the very relevant work by Kant and Thomassen and the anonymous referees
for their helpful and thoughtful comments.
\end{acknowledgements}

\ifArxiv
\bibliographystyle{abuser}
\fi
{
\begin{small}
\vspace{-.3cm}\bibliography{stephen}
\end{small}
}

\ifArxiv
\newpage
\begin{appendix}
\section{Examples of Touching Hexagons Graph Drawings}
\label{sec:examples}

\begin{figure}[hb]
\begin{center}
\includegraphics[width=\textwidth]{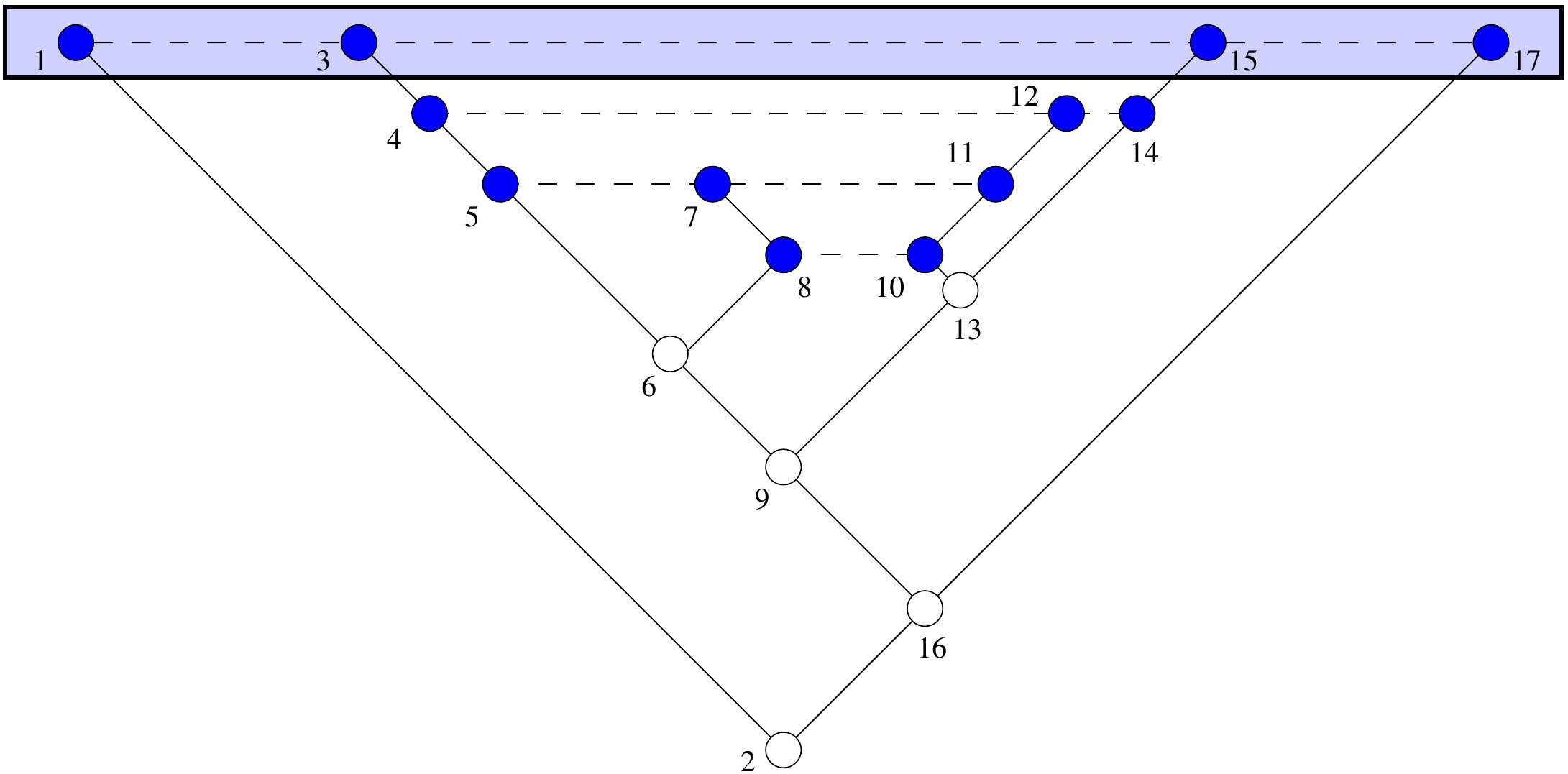}\\
(a)\\
\vspace{0.2cm}
\begin{tabular}{|c|c|c|c|c|c|}
\hline
Node & $x$ & $y$ & shift & h & Steps\\
\hline
1 & 0 & 0 & 0 & 0 & 2a\\
3 & 0 & 0 & 0 & 0 & 2a\\
15 & 0 & 0 & 0 & 0 & 2a\\
17 & 0 & 0 & 0 & 0 & 2a\\
\hline
\end{tabular}\\
\vspace{0.2cm}
(b)
\end{center}
\caption{\small\sf Example execution of the compaction algorithm on 
the graph and subsequent capped binary tree $G_T$
from Figure~\ref{fig:cappedTree}.
(a) The tree with the initial active front set $F$ highlighted.
Observe that 7 and 12 are leaf nodes but are not in $F$ because
$(5,7,11)$ and $(4,12,14)$ are not all leaves in the current working
tree $G_C = G_T$.
(b) The table of the nodes altered in the current iteration
of the algorithm indicating the $x$, $y$, and shift values
for each node along with the specific steps applied to determine
these new values.}
\label{fig-compExample-I}
\end{figure}
\vfill
\begin{figure}
\begin{center}
\begin{tabular}{cc}
\includegraphics[scale=0.5]{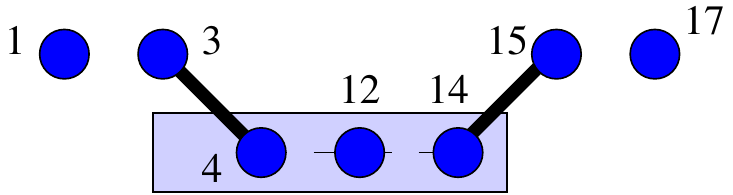} &
\begin{tabular}[b]{|c|c|c|c|c|c|}
\hline
Node & $x$ & $y$ & shift & h & Steps\\
\hline
4 & 1 & -1 & 0 & 1 & 2b,3b\\
12 & 0 & -1 & 0 & 1 & 2a,3a\\
14 & -1 & -1 & 0 & 1 & 2b,3b\\
\hline
\end{tabular}\tabularnewline
\multicolumn{2}{c}{(a)}\\
\\
\includegraphics[scale=.5]{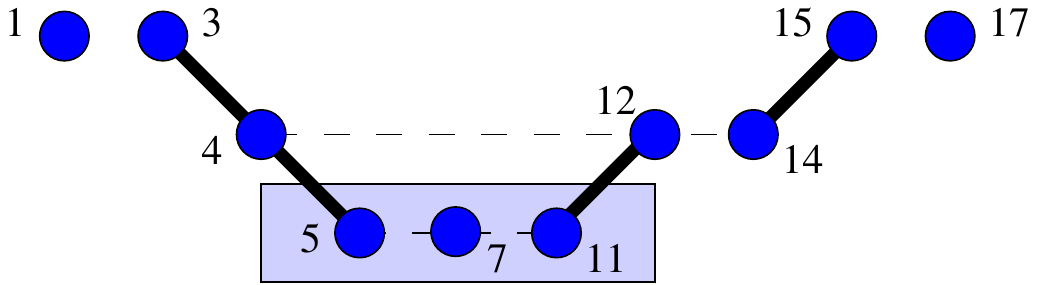} &
\begin{tabular}[b]{|c|c|c|c|c|c|}
\hline
Node & $x$ & $y$ & shift & h & Steps\\
\hline
5 & 2 & -2 & 0 & 2 & 2b,3b\\
7 & 0 & -2 & 0 & 2 & 2a,3a\\
11 & -1 & -2 & 0 & 2 & 2b,3b\\
\hline
\end{tabular}\tabularnewline
\multicolumn{2}{c}{(b)}\\
\\
\includegraphics[scale=.5]{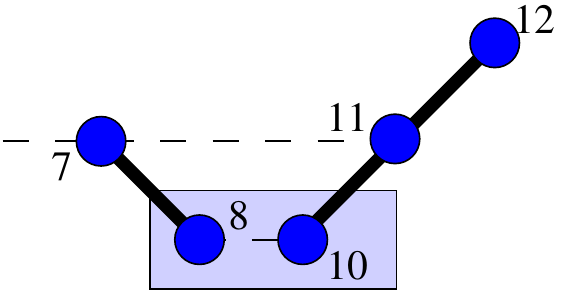} &
\begin{tabular}[b]{|c|c|c|c|c|c|}
\hline
Node & $x$ & $y$ & shift & h & Steps\\
\hline
8 & 1 & -3 & 0 & 3 & 2b,3a\\
10 & -2 & -3 & 0 & 3 & 2b,3a\\
\hline
\end{tabular}\tabularnewline
\multicolumn{2}{c}{(c)}
\end{tabular}
\end{center}
\caption{\small\sf The next three iterations of subsequent active fronts.
For space only a portion of the tree is shown.
(a) The active front consisting of nodes 4, 12, and 14.
Observe that all three nodes are placed (by step 3) at the lowest
$y$-value.  Thus, node 12's $y$ position is set to $-1$.
(b) The active front consisting of nodes 5, 7, and 11.
(c) The active front consisting of nodes 8 and 10.}
\label{fig-compExample-II}
\end{figure}

\begin{figure}
\begin{center}
\begin{tabular}{cc}
\includegraphics[scale=0.5]{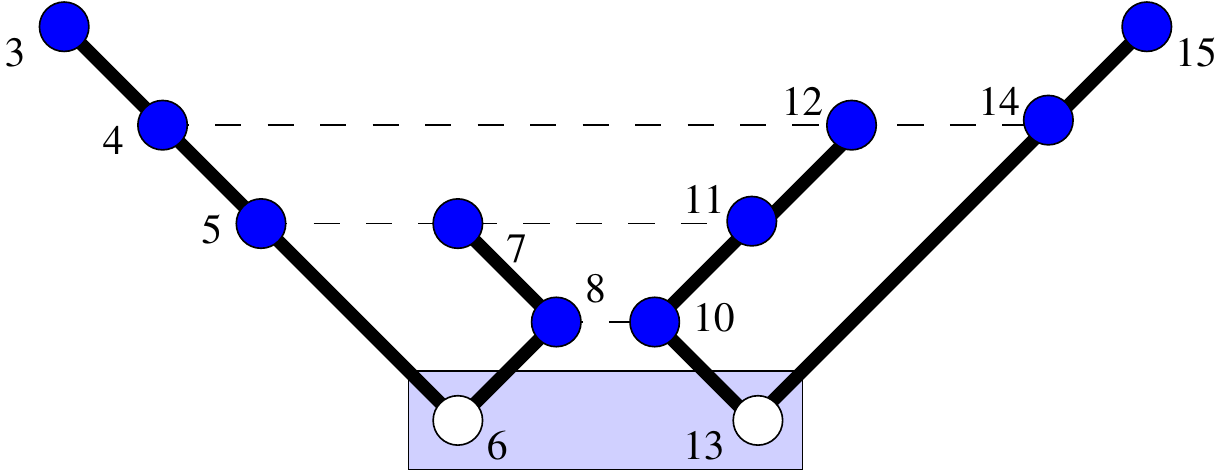} &
\begin{tabular}[b]{|c|c|c|c|c|c|}
\hline
Node & $x$ & $y$ & shift & h & Steps\\
\hline
 6 &  4 & -4 & 0 & - & 2c\\
13 & -1 & -4 & 0 & - & 2c\\
\hline\hline
 8 &  1 & -3 & 4 & - & 2c\\
14 & -1 & -1 & 3 & - & 2c\\
\hline
\end{tabular}\tabularnewline
\multicolumn{2}{c}{(a)}\\
\\
\includegraphics[scale=.5]{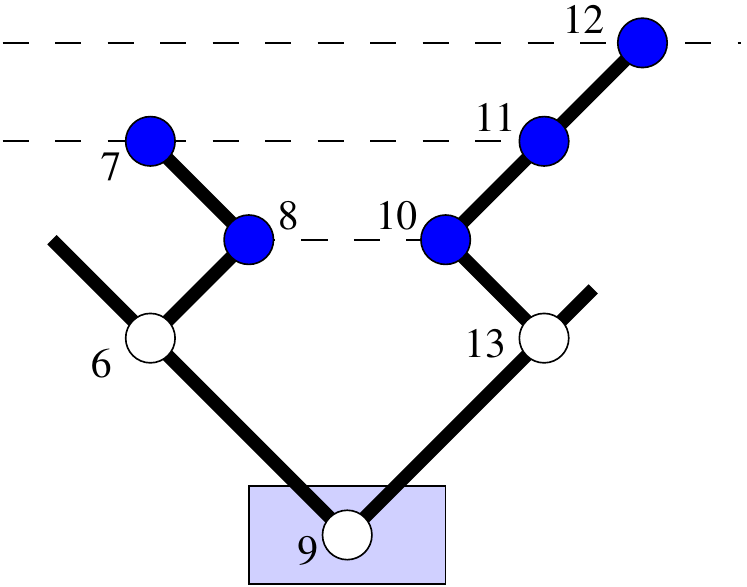} &
\begin{tabular}[b]{|c|c|c|c|c|c|}
\hline
Node & $x$ & $y$ & shift & h & Steps\\
\hline
 9 &  6 & -6 & 0 & - & 2c\\
\hline\hline
13 & -1 & -4 & 9 & - & 2c\\
\hline
\end{tabular}\tabularnewline
\multicolumn{2}{c}{(b)}\\
\\
\includegraphics[scale=.5]{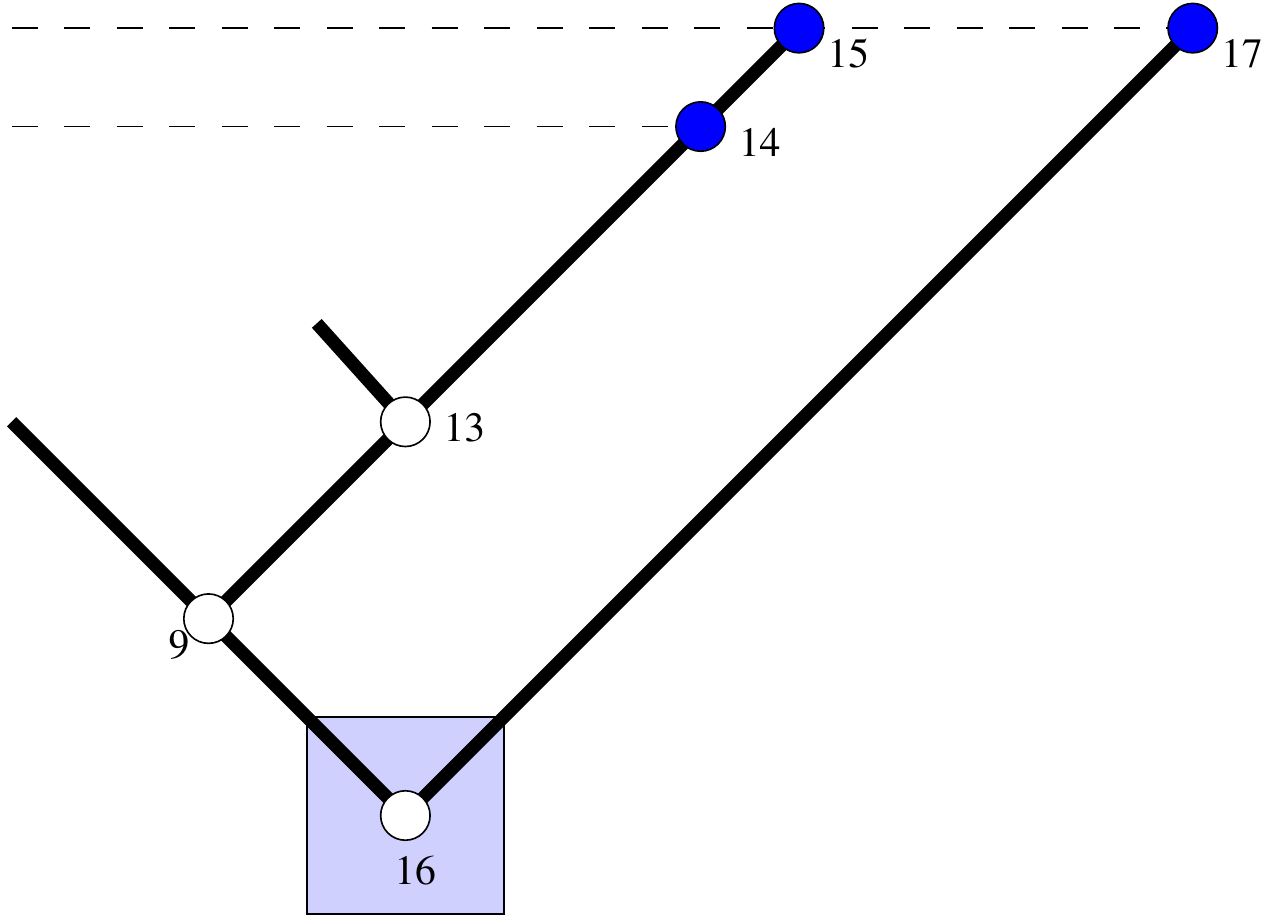} &
\begin{tabular}[b]{|c|c|c|c|c|c|}
\hline
Node & $x$ & $y$ & shift & h & Steps\\
\hline
16 & 7 & -7 &  0 & - & 2c\\
\hline\hline
17 & 0 &  0 & 14 & - & 2c\\
\hline
\end{tabular}\tabularnewline
\multicolumn{2}{c}{(c)}
\end{tabular}
\end{center}
\caption{\small\sf The next three iterations of subsequent active fronts.
In these examples, there are no cap nodes involved so Step 3
is not executed.
(a) The active front consisting of nodes 6 and 13.
Nodes 8 and 14 are affected because they are shifted as
part of the compaction step (2c).
(b) The active front consisting of node 9.
Node 13 is shifted in Step 2c because it is the root of the 
right subtree of node 9.
(c) The active front consisting of node 16.
Node 17 is shifted in Step 2c because it is the root of the
right subtree of 16.
}
\label{fig-compExample-III}
\end{figure}

\vfill

\begin{figure}
\begin{center}
\includegraphics[width=\textwidth]{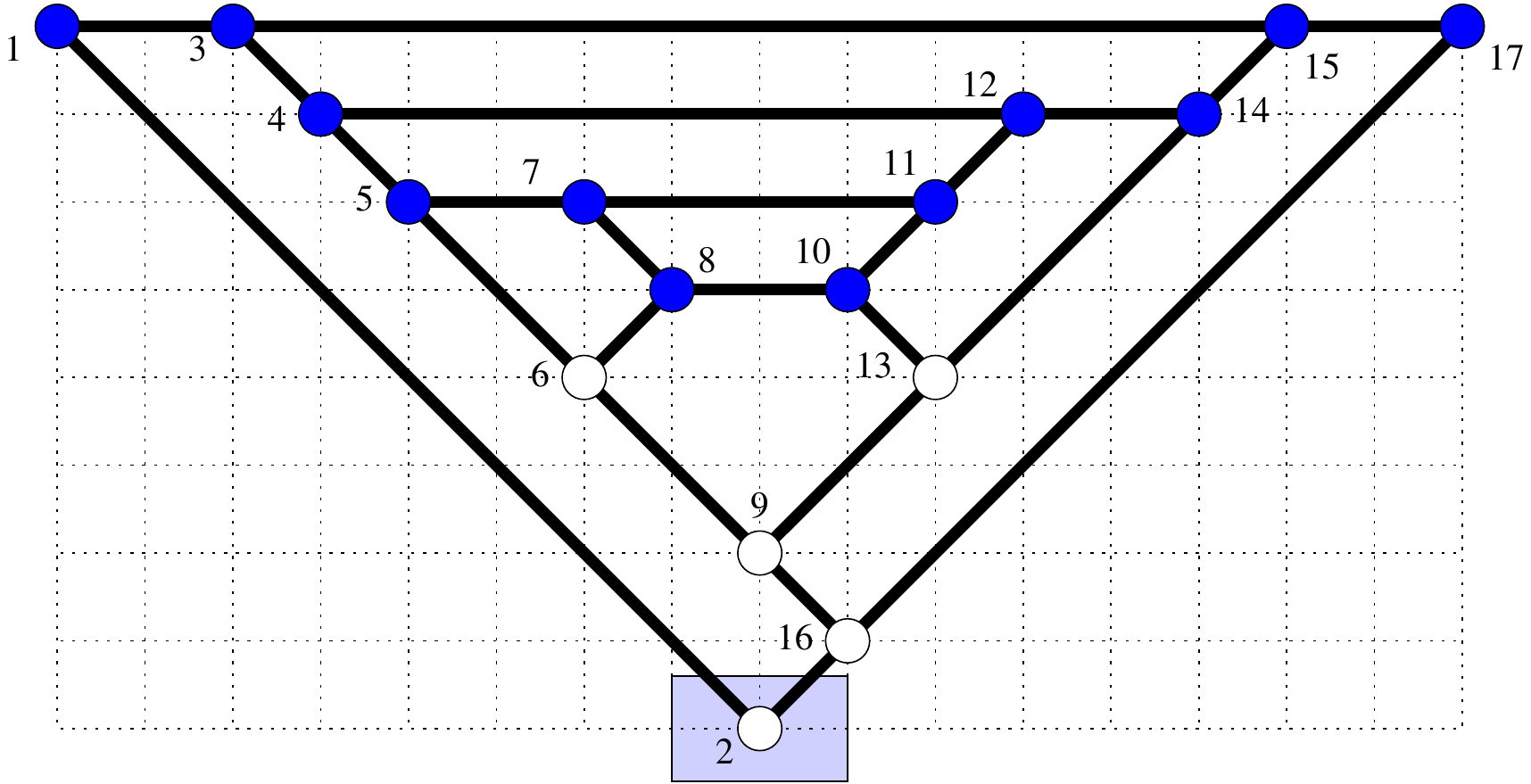}\\
(a)\\
\vspace{0.2cm}
\begin{tabular}{|c|c|c|c|c|c|}
\hline
Node & $x$ & $y$ & shift & h & Steps\\
\hline
 2 & 8 & -8 & 0 & - & 2c\\
\hline\hline
16 & 7 & -7 & 2 & - & 2c\\
\hline
\end{tabular}\\
\vspace{0.2cm}
(b)\\
\vspace{0.2cm}
\begin{tabular}{|c|c|c|c|}
\hline
Node & $x$ & $y$ & shift\\
\hline
 2 &  8 & -8 & 0\\
 1 &  0 &  0 & 0\\
16 &  9 & -7 & 2\\
 9 &  8 & -6 & 2\\
 6 &  6 & -4 & 2\\
 5 &  4 & -2 & 2\\
 4 &  3 & -1 & 2\\
 3 &  2 &  0 & 2\\
 8 &  7 & -3 & 6\\
 7 &  6 & -2 & 6\\
13 & 10 & -4 & 11\\
10 &  9 & -3 & 11\\
11 & 10 & -2 & 11\\
12 & 11 & -1 & 11\\
14 & 13 & -1 & 14\\
15 & 14 &  0 & 14\\
17 & 16 &  0 & 16\\
\hline
\end{tabular}\\
\vspace{0.2cm}
(c)
\end{center}
\caption{\small\sf The final iteration along with the propagation
of the shifts to the children.
(a) The final drawing.
(b) The table of the two affected nodes by the last step.
(c) The final table of all node positions including the
shift values applied after propagation, listed in
preorder traversal of the tree.
}
\label{fig-compExample-IV}
\end{figure}

\begin{figure}[hb]
\begin{center}
\begin{tabular}{cc}
\includegraphics[scale=.3]{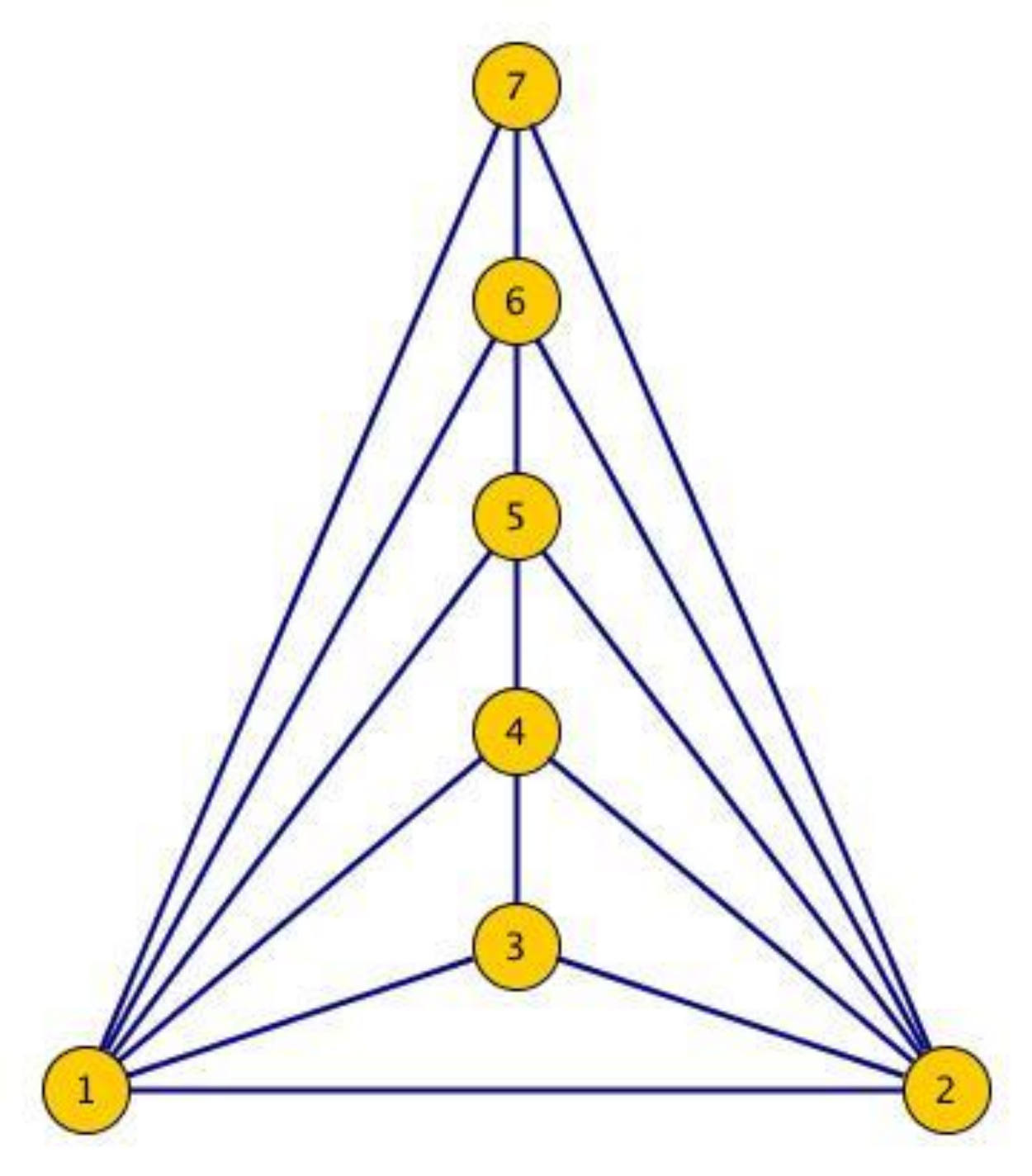} &
\includegraphics[width=.4\textwidth]{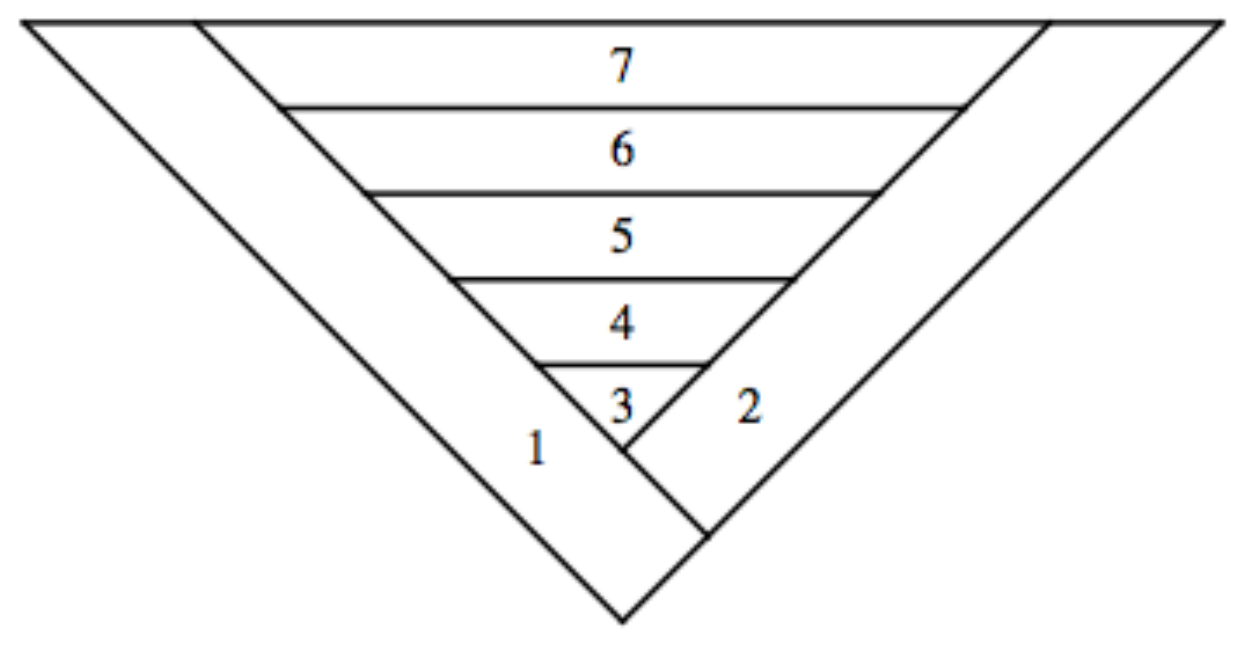}\\
\includegraphics[scale=.3]{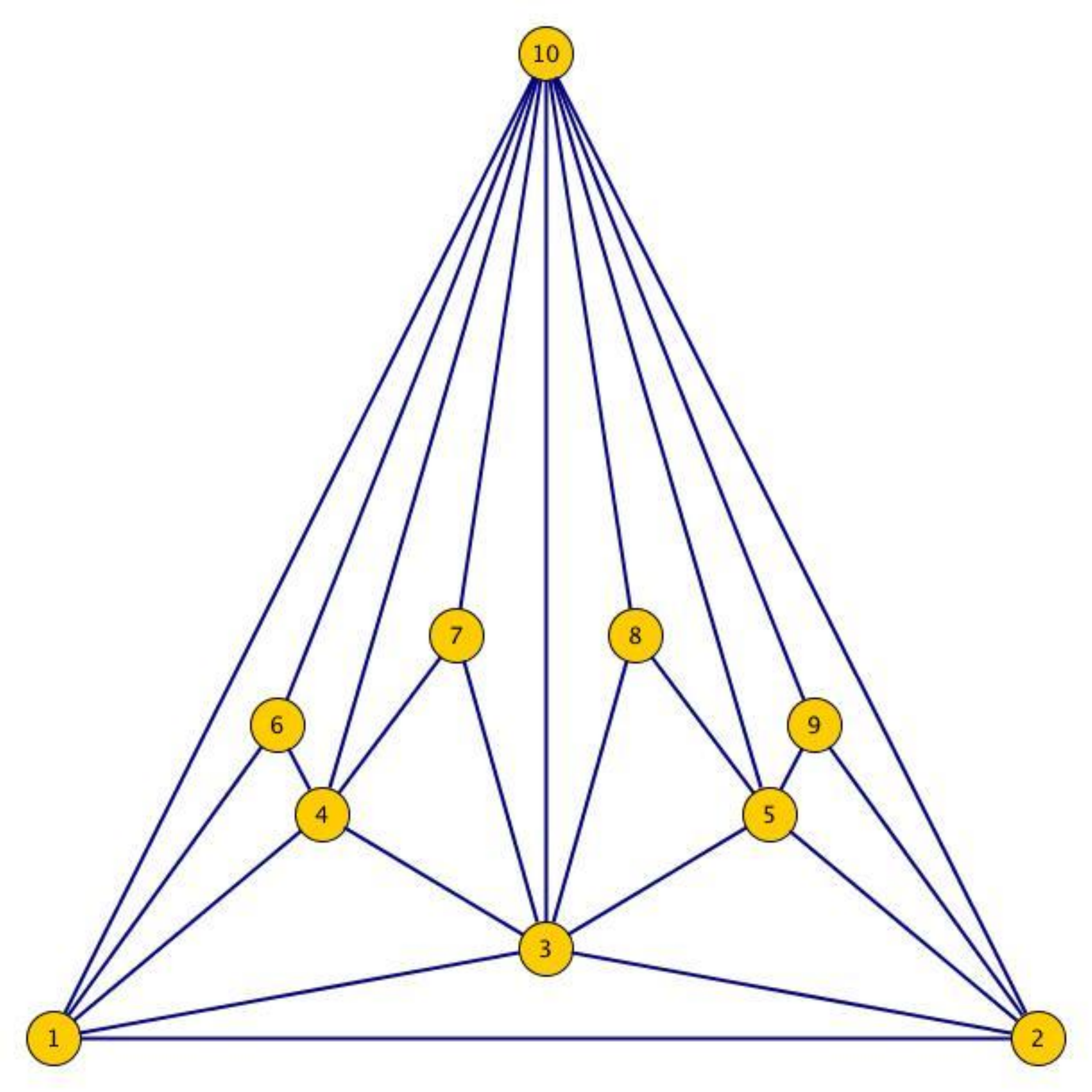} &
\includegraphics[width=.4\textwidth]{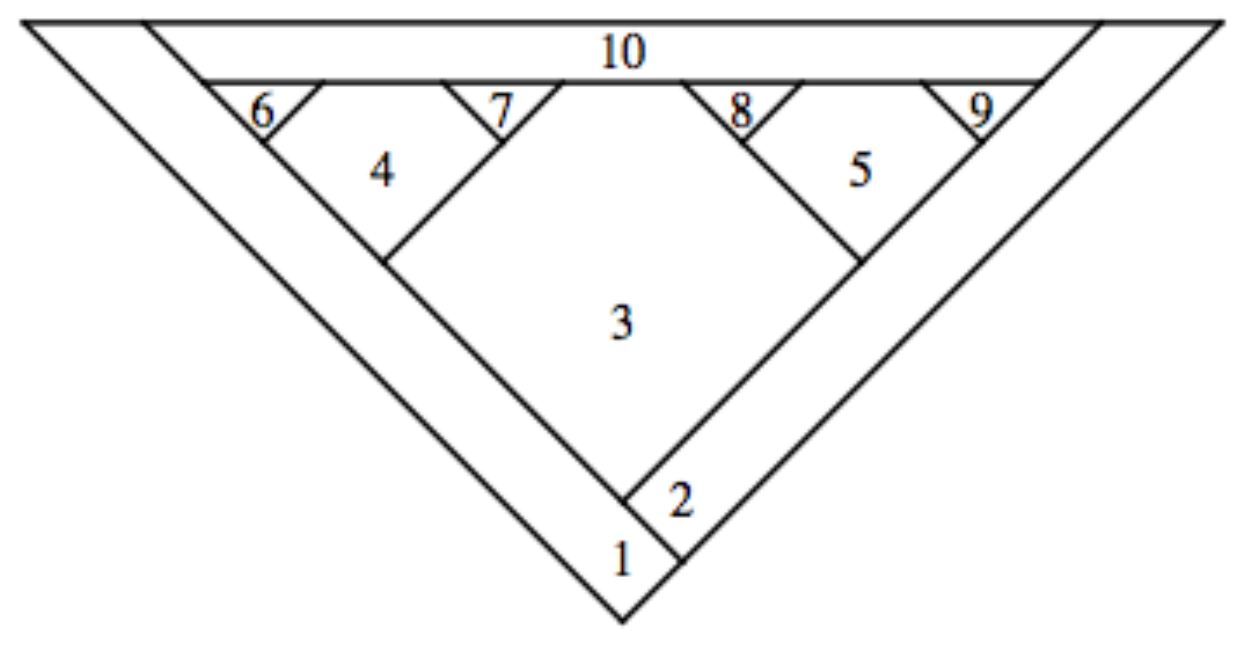}\\
\includegraphics[scale=.3]{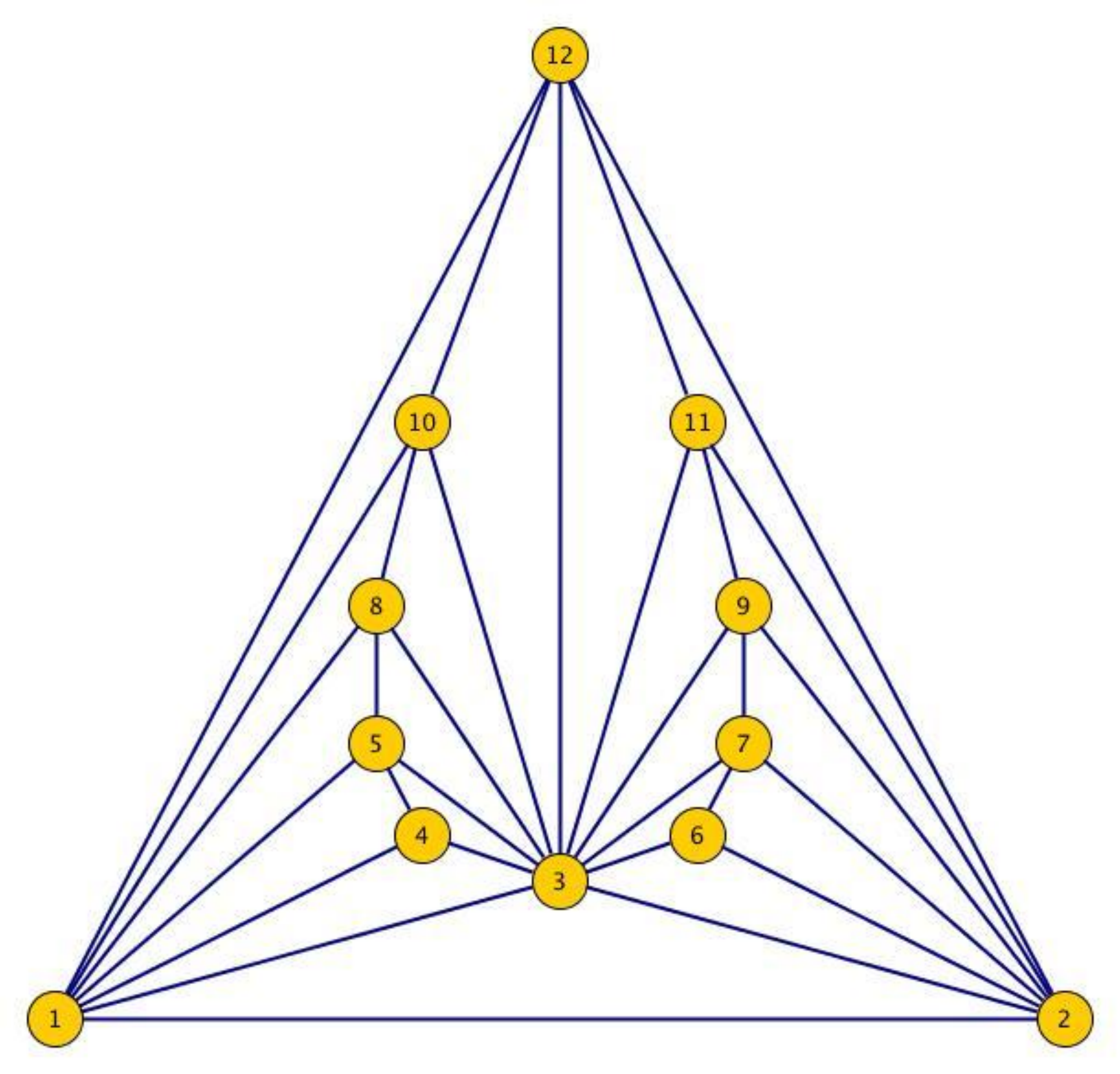} &
\includegraphics[width=.4\textwidth]{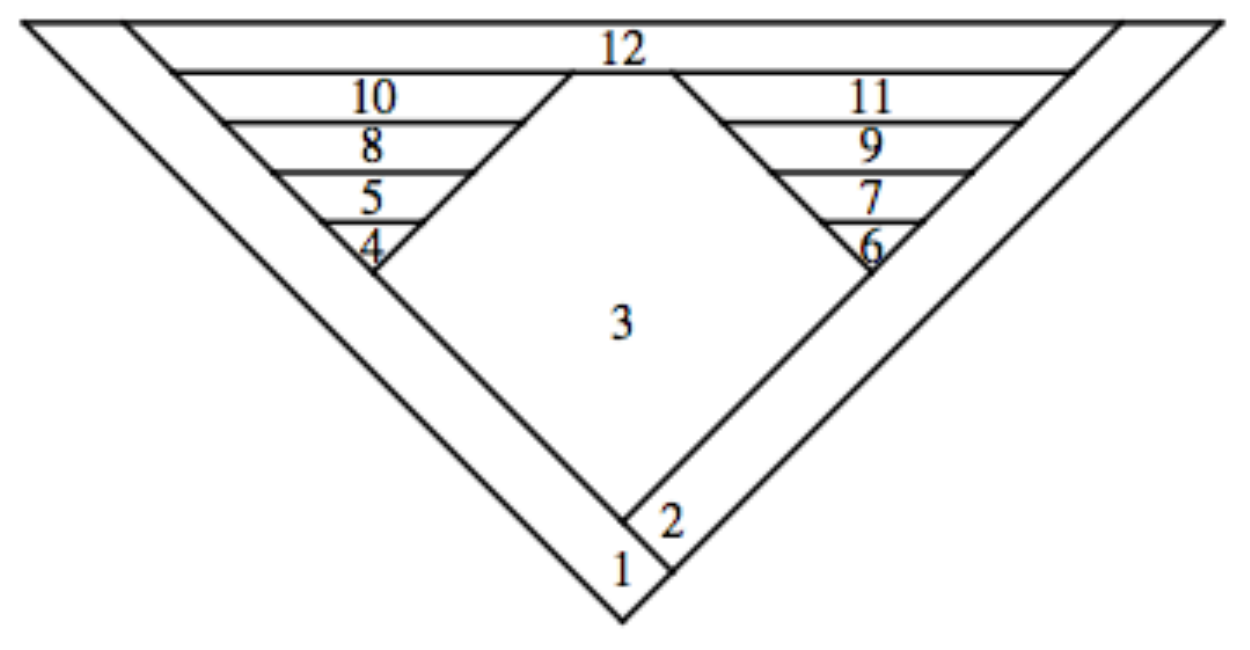}
\end{tabular}
\end{center}
\vspace{-.2cm}\caption{\small\sf Examples illustrating input graphs and their corresponding touching
hexagons representations.}
\label{fig-samples}
\end{figure}

\end{appendix}
\fi
\end{document}